\documentclass[acmsmall, nonacm, anonymous = false, screen]{acmart}
\usepackage{amsthm,latexsym,stmaryrd}
\usepackage{dsfont}
\usepackage{thmtools,thm-restate}

\newcommand{\grammar}{\ensuremath{\Lambda}}
\newcommand{\dee}{\mathrm{d}}

\title{Behavioural pseudometrics for continuous-time diffusions}
\author{Linan Chen}
\affiliation{
\institution{McGill University}
\department{Department of Mathematics and Statistics}
\city{Montreal}
\country{Canada}
}
\author{Florence Clerc}
\affiliation{
\institution{McGill University}
\department{School of Computer Science}
\city{Montreal}
\country{Canada}
}
\author{Prakash Panangaden}
\affiliation{
\institution{McGill University}
\department{School of Computer Science}
\city{Montreal}
\country{Canada}
}
\affiliation{
\institution{MILA}
\city{Montreal}
\country{Canada}
}

\makeatletter
\newcommand{\singlespacing}{\let\CS=\@currsize\renewcommand{\baselinestretch}{1}\small\CS}
\newcommand{\doublespacing}{\let\CS=\@currsize\renewcommand{\baselinestretch}{1.75}\small\CS}
\newcommand{\normalspacing}{\let\CS=\@currsize\renewcommand{\baselinestretch}{\BLS}\small\CS}
\makeatother

\theoremstyle{definition}
\newtheorem{thm}{Theorem}[section]

\newtheorem{cor}[thm]{Corollary}

\newtheorem{prop}[thm]{Proposition}

\newtheorem{remark}[thm]{Remark}
 
\newtheorem{theorem}[thm]{Theorem}

\newtheorem{lemma}[thm]{Lemma}
\newtheorem{proposition}[thm]{Proposition}

\newtheorem{definition}[thm]{Definition}
\newtheorem{example}[thm]{Example}

\newtheorem*{cor*}{Corollary}

\newcommand{\x}{\times}

 \def\pushright#1{{
    \parfillskip=0pt            
    \widowpenalty=10000         
    \displaywidowpenalty=10000  
    \finalhyphendemerits=0      
    \leavevmode                 
    \unskip                     
    \nobreak                    
    \hfil                       
    \penalty50                  
    \hskip.2em                  
    \null                       
    \hfill                      
    {#1}                        
    \par}}                      
 \def\qed{\pushright{\rule{2mm}{3mm}}\penalty-700 \smallskip}

\renewenvironment{proof}{\begin{trivlist} \item[{\bf ~Proof}.]}%
 {\qed\end{trivlist}}

\makeatletter

\newdimen\w@dth

\def\setw@dth#1#2{\setbox\z@\hbox{\scriptsize $#1$}\w@dth=\wd\z@
\setbox\@ne\hbox{\scriptsize $#2$}\ifnum\w@dth<\wd\@ne \w@dth=\wd\@ne \fi
\advance\w@dth by 1.2em}

\def\t@^#1_#2{\allowbreak\def\n@one{#1}\def\n@two{#2}\mathrel
{\setw@dth{#1}{#2}
\mathop{\hbox to \w@dth{\rightarrowfill}}\limits
\ifx\n@one\empty\else ^{\box\z@}\fi
\ifx\n@two\empty\else _{\box\@ne}\fi}}
\def\t@@^#1{\@ifnextchar_ {\t@^{#1}}{\t@^{#1}_{}}}

\def\t@left^#1_#2{\def\n@one{#1}\def\n@two{#2}\mathrel{\setw@dth{#1}{#2}
\mathop{\hbox to \w@dth{\leftarrowfill}}\limits
\ifx\n@one\empty\else ^{\box\z@}\fi
\ifx\n@two\empty\else _{\box\@ne}\fi}}
\def\t@@left^#1{\@ifnextchar_ {\t@left^{#1}}{\t@left^{#1}_{}}}

\def\two@^#1_#2{\def\n@one{#1}\def\n@two{#2}\mathrel{\setw@dth{#1}{#2}
\mathop{\vcenter{\hbox to \w@dth{\rightarrowfill}\kern-1.7ex
                 \hbox to \w@dth{\rightarrowfill}}%
       }\limits
\ifx\n@one\empty\else ^{\box\z@}\fi
\ifx\n@two\empty\else _{\box\@ne}\fi}}
\def\tw@@^#1{\@ifnextchar_ {\two@^{#1}}{\two@^{#1}_{}}}

\def\tofr@^#1_#2{\def\n@one{#1}\def\n@two{#2}\mathrel{\setw@dth{#1}{#2}
\mathop{\vcenter{\hbox to \w@dth{\rightarrowfill}\kern-1.7ex
                 \hbox to \w@dth{\leftarrowfill}}%
       }\limits
\ifx\n@one\empty\else ^{\box\z@}\fi
\ifx\n@two\empty\else _{\box\@ne}\fi}}
\def\t@fr@^#1{\@ifnextchar_ {\tofr@^{#1}}{\tofr@^{#1}_{}}}

\newdimen\W@dth
\def\setW@dth#1#2{\setbox\z@\hbox{$#1$}\W@dth=\wd\z@
\setbox\@ne\hbox{$#2$}\ifnum\W@dth<\wd\@ne \W@dth=\wd\@ne \fi
\advance\W@dth by 1.2em}

\def\T@^#1_#2{\allowbreak\def\N@one{#1}\def\N@two{#2}\mathrel
{\setW@dth{#1}{#2}
\mathop{\hbox to \W@dth{\rightarrowfill}}\limits
\ifx\N@one\empty\else ^{\box\z@}\fi
\ifx\N@two\empty\else _{\box\@ne}\fi}}
\def\T@@^#1{\@ifnextchar_ {\T@^{#1}}{\T@^{#1}_{}}}

\def\T@left^#1_#2{\def\N@one{#1}\def\N@two{#2}\mathrel{\setW@dth{#1}{#2}
\mathop{\hbox to \W@dth{\leftarrowfill}}\limits
\ifx\N@one\empty\else ^{\box\z@}\fi
\ifx\N@two\empty\else _{\box\@ne}\fi}}
\def\T@@left^#1{\@ifnextchar_ {\T@left^{#1}}{\T@left^{#1}_{}}}

\def\Tofr@^#1_#2{\def\N@one{#1}\def\N@two{#2}\mathrel{\setW@dth{#1}{#2}
\mathop{\vcenter{\hbox to \W@dth{\rightarrowfill}\kern-1.7ex
                 \hbox to \W@dth{\leftarrowfill}}%
       }\limits
\ifx\N@one\empty\else ^{\box\z@}\fi
\ifx\N@two\empty\else _{\box\@ne}\fi}}
\def\T@fr@^#1{\@ifnextchar_ {\Tofr@^{#1}}{\Tofr@^{#1}_{}}}

\def\Two@^#1_#2{\def\N@one{#1}\def\N@two{#2}\mathrel{\setW@dth{#1}{#2}
\mathop{\vcenter{\hbox to \W@dth{\rightarrowfill}\kern-1.7ex
                 \hbox to \W@dth{\rightarrowfill}}%
       }\limits
\ifx\N@one\empty\else ^{\box\z@}\fi
\ifx\N@two\empty\else _{\box\@ne}\fi}}
\def\Tw@@^#1{\@ifnextchar_ {\Two@^{#1}}{\Two@^{#1}_{}}}

\def\to{\@ifnextchar^ {\t@@}{\t@@^{}}}
\def\from{\@ifnextchar^ {\t@@left}{\t@@left^{}}}
\def\two{\@ifnextchar^ {\tw@@}{\tw@@^{}}}
\def\tofro{\@ifnextchar^ {\t@fr@}{\t@fr@^{}}}
\def\To{\@ifnextchar^ {\T@@}{\T@@^{}}}
\def\From{\@ifnextchar^ {\T@@left}{\T@@left^{}}}
\def\Two{\@ifnextchar^ {\Tw@@}{\Tw@@^{}}}
\def\Tofro{\@ifnextchar^ {\T@fr@}{\T@fr@^{}}}

\makeatother

\newcommand{\cA}{\mathcal{A}}

\newcommand{\cC}{\mathcal{C}}

\newcommand{\cE}{\mathcal{E}}
\newcommand{\cF}{\mathcal{F}}
\newcommand{\cG}{\mathcal{G}}
\newcommand{\cH}{\mathcal{H}}
\newcommand{\cI}{\mathcal{I}}

\newcommand{\cL}{\mathcal{L}}
\newcommand{\cM}{\mathcal{M}}

\newcommand{\cO}{\mathcal{O}}
\newcommand{\cP}{\mathcal{P}}

\newcommand{\cT}{\mathcal{T}}

\newcommand{\cX}{\mathcal{X}}

\begin{document}

\begin{abstract}
   Bisimulation is a concept that captures behavioural equivalence of states in a variety
  of types of transition systems.  It has been widely studied in a discrete-time setting
  where the notion of a step is fundamental.  In our setting we are considering  ``flow''-processes emphasizing that they 
 evolve in continuous time.  In such continuous-time settings, the
  concepts are not straightforward adaptations of their
  discrete-time analogues and we restrict our study to  diffusions that do not lose mass over time and with additional regularity constraints.
  
  In \cite{Chen19a, Chen20}, Chen et al. proposed different definitions of
  behavioural equivalences for continuous-time stochastic processes where the evolution is
  a \emph{flow} through time.  That work only addressed equivalen\-ces. 
  In this work, we aim at quantifying how differently processes behave. We present two pseudometrics for diffusion-like processes.
  These pseudometrics are 
  fixpoints of two different functionals on the space of 1-bounded pseudometrics on the state space.  We
  also characterize these pseudometrics in terms of real-valued modal logics; this is a
  quantitative analogue of the notion of logical characte\-rization of bisimulation. These real-valued modal logics indicate that the two pseudometrics are different and thus yield different notions of behavioural equivalence.
\end{abstract}

\maketitle

\section{Introduction}

Bisimulation~\cite{Milner80,Park81,Sangiorgi09} is a fundamental concept in the theory of
transition systems capturing a strong notion of behavioural equivalence.  The extension to
probabilistic systems is due to Larsen and Skou~\cite{Larsen91}; henceforth we will simply
say ``bisimulation'' instead of ``probabilistic bisimulation''.  Bisimulation has been
studied for discrete-time systems where transitions happen as steps, both on
discrete~\cite{Larsen91} and continuous state
spaces~\cite{Blute97,Desharnais98,Desharnais02}.  In all these types of systems, a crucial
ingredient of the definition of bisimulation is the ability to talk about \emph{the next
  step}.  This notion of bisimulation is characterized by a modal logic~\cite{Larsen91}
even when the state space is continuous~\cite{Desharnais98}.

Some work had previously been done in what are called continuous-time systems, see for
example~\cite{Baier08}, but even in so-called continuous-time Markov chains there is a
discrete notion of time \emph{step}; it is only that there is a real-valued duration
associated with each state that leads to people calling such systems continuous-time.
They are often called ``jump processes'' in the mathematical literature (see, for example,
\cite{Rogers00a,Whitt02}), a phrase that better captures the true nature of such
processes.  Metrics and equivalences for such processes were studied by Gupta et
al.~\cite{Gupta04,Gupta06} but the results therein do not address the subtlety in our problem (see Remark \ref{rem:subtlety}).

There is a vast range of systems that involve true continuous-time evolution:
deterministic systems governed by differential equations and stochastic systems governed
by ``noisy'' differential equations called stochastic differential equations.  These have
been extensively studied for over a century since the pioneering work of
Einstein~\cite{Einstein1905} on Brownian motion.  In the physics literature such systems
are often studied through partial differential equations governing the evolution of the
probability density function; the equations are typically known as the Kolmogorov backward equation or the Fokker-Planck equation.

The processes we consider have continuous state spaces and are governed by a continuous
time evolution, a paradigmatic example is Brownian motion.  We will use the adjective
``flow'' to emphasize the distinction with jump processes.  A natural, almost naive, idea is to view a continuous-time process being ``approximated'' by a discrete-time process and then
one can take, in some proper sense,  a limit as the step size goes to zero.  However, entirely new
phenomena and difficulties manifest themselves in this procedure.  For example, even the basic properties of trajectories of Brownian motion are vastly more complicated than the counterparts of a random walk.  Basic concepts like ``the time at which a process exits a given
subset of the state space'' becomes intricate to define.  Notions like ``matching transition
steps'' are no longer applicable as the notion of ``step'' does not make sense. 

In~\cite{Chen19a, Chen20, Chen23}, Chen et al. proposed different notions of behavioural
equivalences on continuous-time processes.
They showed that there were several possible extensions of the notion of bisimulation to continuous time and that the continuous-time notions needed to involve trajectories in order to be meaningful.  There were significant mathematical challenges in even proving that an equivalence relation existed.  For example,  obstacles occurred in establishing
measurability of various functions and sets, due to the
inability to countably generate the relevant $\sigma$-algebras.  Those papers
left completely open the question of defining a suitable pseudometric analogue, a concept
that would be more useful in practice than an equivalence relation.

Previous work on discrete-time Markov processes by Desharnais et al.~\cite{Desharnais99b, Desharnais04}
extended the modal logic characterizing bisimulation to a real-valued logic that allowed
to not only state if two states were ``behaviourally equivalent'' but, more interestingly, how similarly they behaved.  This shifts the notion from a qualitative notion (an
equivalence) to a quantitative one (a pseudometric).

Other work also on discrete-time Markov processes by van Breugel et
al.~\cite{vanBreugel05a} introduced a slightly different real-valued logic and compared
the corresponding pseudometric to another pseudometric obtained as a terminal coalgebra of a
carefully crafted functor.

Our goal is to quantify the difference between behaviours of \emph{flow} (continuous-time) processes.  The behaviour of a process is given by the \emph{dynamics} of the system.  For continuous-time systems, there are two ways of describing the dynamics: either through a time-indexed  family of  Markov kernels or through a space-indexed family of probability distributions on the space of trajectories.  This leads to  two different pseudometrics each defined in two different ways.
\begin{enumerate}
\item The first way of defining a pseudometric is through a functional.
More
specifically,  given a pseudometric, we view it as some cost function and define the functional as 
the optimal transport cost between the distributions representing the dynamics of the system.  Then, the desired behavioural pseudometric is a fixpoint of the functional (though it is not defined directly in that way).
\item We then show that this pseudometric is characterized by a real-valued logic that closely resembles the one introduced in
\cite{vanBreugel05a} for discrete-time systems:  this is a quantitative analogue of the logical characterization of
bisimulation.
\end{enumerate}


Generally speaking, we are following a familiar path from equiva\-lences to
logics to metrics.  However, it is necessary for us to redevelop the framework and the mathematical
techniques from scratch.  Compared with the step-based ``jump processes'', these ``flow
processes'' are intrinsically different and have more involved properties, which poses new
challenges. For example, unlike in the discrete cases, certain functionals on
trajectories that we are interested in studying may not even be measurable. To overcome
such difficulties, we will restrict to diffusions with additional regularity conditions.
This restricts the
scope of processes considered compared to previous work by Chen et al~\cite{Chen19a, Chen20, Chen23}.
However these additional regularity conditions are satisfied by many processes commonly used in modeling physic quantities,  such as Brownian motion, geometric Brownian motion, Ornstein-Uhlenbeck process, etc. These conditions allow us to use transport theory which plays a key role in this
study. The whole machinery is quite heavy.
  
\paragraph*{Organization of the paper: } We will first go through some specific mathematical background in Section \ref{sec:background}. In Section \ref{sec:CTsystems}, we will describe the processes that we will study by first defining Feller-Dynkin processes and then what we mean by diffusions.  In Section \ref{sec:functionals-metric}, we will define the first functional using the Markov kernels and the metrics as iterations of said functional. We will then prove that this pseudometric is indeed a fixpoint of our functional. In Section \ref{sec:logics}, we show that this pseudometric corresponds to some real-valued logic.  In Section \ref{sec:diffusions}, we will define a second functional based on trajectories and the corresponding pseudometric. We will further show that it corresponds to a different logic. 

\section{Mathematical background}
\label{sec:background}

We assume the reader to be familiar with basic measure
theory and topology. Nevertheless we provide a brief review of the relevant notions and theorems. Let us start with clarifying a few notations on integrals: Given a measurable space $X$ equipped with a measure $\mu$ and a measurable function $f: X \to \mathbb{R}$, we can write either $\int f~ \dee \mu$ or $\int f(x) ~ \mu (\dee x)$ interchangeably. The second notation will be especially useful when $\mu$ is a Markov kernel $P_t(x)$ for some $t \geq 0$ and $x \in X$: $\int f(y) ~P_t(x, \dee y) = \int f ~ \dee P_t(x)$.

All the proofs for this Section can be found in Appendix \ref{sec:proof-sec-background}.

   \subsection{Lower semi-continuity}

\begin{definition}
Given a topological space $X$, a function $f: X\to \mathbb{R}$ is \emph{lower semi-continuous} if for every $x_0 \in X$, $\liminf_{x \rightarrow x_0} f(x) \geq f(x_0)$. This condition is equivalent to the following one: for any $y \in \mathbb{R}$, the set $f^{-1}((y, + \infty)) = \{ x ~|~ f(x) > y\}$ is open in $X$.
\end{definition}

\begin{restatable}{lemma}{lemmasuplowersemicontinuous}
\label{lemma:sup-lower-semicontinuous}
Assume there is an arbitrary family of continuous functions $f_i : X \to \mathbb{R}$ ($i \in \cI$) and define $f(x) = \sup_{i \in \cI} f_i(x)$ for every $x \in X$.  Then $f$ is lower semi-continuous.
\end{restatable}

The converse is also true, as Baire's theorem states:
\begin{thm}
\label{thm:baire}
If $X$ is a metric space and if a function $f: X \to \mathbb{R}$ is lower semi-continuous,  then $f$ is the limit of an increasing sequence of real-valued continuous functions on $X$.
\end{thm}

    \subsection{Couplings}
\label{sec:coupling}

\begin{definition}
Let $(X, \Sigma_X, P)$ and $(Y, \Sigma_Y, Q)$ be two probability spaces.  Then a \emph{coupling} $\gamma$ of $P$ and $Q$ is a probability distribution on $(X \times Y, \Sigma_X \otimes \Sigma_Y)$ such that for every  $B_X \in \Sigma_X$, $\gamma(B_X \times Y) = P(B_X)$ and for every $B_Y \in \Sigma_Y$, $\gamma(X \times B_Y) = Q(B_Y)$ ($P, Q$ are called the \emph{marginals} of $\gamma$).  We write $\Gamma(P, Q)$ for the set of couplings of $P$ and~$Q$.
\end{definition}

\begin{restatable}{lemma}{lemmacouplingscompact}
\label{lemma:couplings-compact}
Given two probability measures $P$ and $Q$ on Polish spaces $X$ and $Y$ respectively,  the set of
couplings $\Gamma (P, Q)$ is compact under the topology of weak convergence. 
\end{restatable}

   \subsection{Optimal transport theory}

A lot of this work is based on optimal transport theory. This whole subsection is based on \cite{Villani08} and will be adapted for our framework.

Consider a Polish space $\cX$ and a lower semi-continuous cost function $c: \cX \times \cX \to [0,1]$ such that for every $x \in \cX$, $c(x,x) = 0$.

For every two probability distributions $\mu$ and $\nu$ on $\cX$, we write $W(c)(\mu, \nu)$ for the optimal transport cost from $\mu$ to $\nu$.
The Kantorovich duality states that (see Theorem 5.10(iii) of \cite{Villani08}),
\[W(c)(\mu, \nu) =  \min_{\gamma \in \Gamma(\mu, \nu)} \int c ~ \dee \gamma
= \max_{h \in \cH(c)} \left| \int h ~ \dee \mu - \int h ~\dee \nu  \right| \]
where $ \cH(c) = \{ h: \cX \to [0,1] ~|~ \forall x,y ~~ |h(x) - h(y)| \leq c(x,y) \}$.

\begin{remark}
\label{rem:dual-pseudometric}
Note that it is not the exact expression in Theorem 5.10 of \cite{Villani08} but the one found in Particular Case 5.16.  The former expression also applies to our case.  Indeed, according to Theorem 5.10, the dual expression is
\[
\min_{\gamma \in \Gamma(\mu, \nu)} \int c ~ \dee \gamma
= \max_{\phi, \psi} \left( \int \phi ~ \dee  \mu - \int \psi ~ \dee \nu \right) \]
where $\phi$ and $\psi$ are such that  $\forall x, y~ |\phi(x) - \psi(y)| \leq c(x,y)$.
However, since $c(x,x) = 0$ for every $x \in \cX$, then for any pair of functions $\phi$ and $\psi$ considered, for all $x, \in X$,
$ |\phi(x) - \psi(x)| \leq c(x,x) = 0,$ 
which implies that $\phi = \psi$.
\end{remark}
   
\begin{restatable}{lemma}{lemmaoptimaltransportpseudometric}
\label{lemma:optimal-transport-pseudometric}
If the cost function $c$ is a 1-bounded  pseudometric on $\cX$, then $W(c)$ is a 1-bounded pseudometric on the space of probability distributions on $\cX$.
\end{restatable}

We will later need the following technical lemma. Theorem 5.20 of \cite{Villani08} states that a sequence $W(c_k)(P_k, Q_k)$ converges to $W(c)(P, Q)$ if $c_k$ uniformly converges to $c$ and $P_k$ and $Q_k$ converge weakly to $P$ and $Q$ respectively. Uniform convergence in the cost function may be too strong a condition for us, but the following lemma is enough for what we need.

\begin{restatable}{lemma}{lemmawassersteinlimit}
\label{lemma:wasserstein-limit}
Consider a Polish space $\cX$ and a cost function $c: \cX \times \cX \to [0,1]$ such that there exists an increasing ($c_{k+1} \geq c_k$ for every $k$) sequence of continuous cost functions $c_k: \cX \times \cX \to [0,1]$ that converges to $c$ pointwise.  Then, given two probability distributions $P$ and $Q$ on $\cX$,
\[ \lim_{k \rightarrow \infty} W(c_k)(P,Q) = W(c)(P, Q). \]
\end{restatable}

\section{The processes we consider: diffusions}
\label{sec:CTsystems}

This work focuses on honest diffusions, continuous-time processes without loss of mass over time and with additional regularity conditions. Before formally defining diffusions, one needs first to define continuous-time Markov processes.
A typical example is Feller-Dynkin processes.  The additional regularity assump\-tions on diffusions give us more leverage to ``cope'' with continuous time and the many technical difficulties that come with it. Much of this material is adapted from \cite{Rogers00a} and we use
their notations.  Another useful source is \cite{Bobrowski05}.  

   \subsection{Definition of Feller-Dynkin Process}
   \label{app:FDP}

Let $E$ be a locally compact, Hausdorff
space with a countable base.  We also equip the set $E$ with its Borel $\sigma$-algebra $\mathcal{B}(E)$ that we will denote by $\cE$.  The previous topological hypotheses also imply that $E$ is $\sigma$-compact and Polish.
We will denote $\Delta$ for the 1-bounded metric that generates the topology making $E$ Polish.

\begin{definition}
A \emph{semigroup} of operators on any Banach space $X$ is a family
of linear continuous (bounded) operators $\cP_t: X \to X$ indexed by
$t\in\mathbb{R}_{\geq 0}$ such that
\[ \forall s,t \geq 0, \cP_s \circ \cP_t = \cP_{s+t} \qquad \text{(semigroup property)}\]
and
\[ \cP_0 = I \qquad \text{(the identity)}.  \]
\end{definition}
\begin{definition}
For $X$ a Banach space, we say that a semigroup $\cP_t:X\to X$ is \emph{strongly continuous} if 
\[ \forall x\in X, \lim_{t\downarrow 0}\| \cP_t x - x \| \to 0.  \]
\end{definition}


What the semigroup property expresses is that we do not need to understand the past (what
happens before time $t$) in order to compute the future (what happens after some
additional time $s$, so at time $t+s$) as long as we know the present (at time
$t$). 

We say that a continuous real-valued function $f$ on $E$ ``vanishes at
infinity'' if for every $\varepsilon > 0$ there is a compact subset
$K \subseteq E$ such that for every $x\in E\setminus K$, we have
$|f(x)| \leq \varepsilon$.  To give an intuition, if $E$ is the real line, this means that
$\lim_{x \rightarrow \pm \infty} f(x) = 0$.  The space $C_0(E)$ of continuous real-valued functions
that vanish at infinity is a Banach space with the $\sup$ norm.   
\begin{definition}
A \emph{Feller-Dynkin (FD) semigroup} is a strongly continuous
semigroup $(\hat{P}_t)_{t \geq 0}$ of linear operators on $C_0(E)$ satisfying the
additional condition:  
\[\forall t \geq 0 ~~~ \forall f \in C_0(E) \text{, if }~~ 0 \leq f \leq 1 \text{, then }~~ 0 \leq \hat{P}_t f \leq 1\]
\end{definition}

The Riesz representation theorem can be found as Theorem II.80.3 of \cite{Rogers00a}.  From it, we can derive the following important proposition which relates these FD-semigroups with Markov
kernels (see Appendix \ref{app:CTsystems} for the details).  This allows one to see the connection with familiar
probabilistic transition systems.  

\begin{proposition}
\label{prop:Riesz-use}
Given an FD-semigroup $(\hat{P}_t)_{t \geq 0}$ on $C_0(E)$, it is possible to define a
unique family of sub-Markov 
kernels $(P_t)_{t \geq 0} : E \times \mathcal{E} \to [0,1]$ such that for
all $t \geq 0$ and $f \in C_0(E)$, 
\[ \hat{P}_t f(x) = \int f(y) P_t(x, \dee y).  \]
\end{proposition}

Given a time $t$ and a state $x$, we will often write $P_t(x)$ for the measure $P_t(x, \cdot)$ on $E$. Note that since $E$ is Polish, then $P_t(x)$ is tight.

\subsection{Trajectories}
\label{sec:trajectories}

A very important ingredient in the theory is the space of trajectories of a
FD-process (FD-semigroup) as a probability space.  This space does not
appear explicitly in the study of labelled Markov processes but one does
see it in the study of continuous-time Markov chains and jump processes.

We will use trajectories in Section \ref{sec:diffusions}. We will impose further conditions (see Definition \ref{def:diffusion}), but for now we follow \cite{Rogers00a} to introduce the concept of trajectories and the corresponding probabilities.

The standard topological hypotheses on the space $E$ allow us to perform the one-point compactification of the set $E$ by adding the absorbing state $\partial$.  We write $E_\partial$ for the corresponding set: $E_\partial = E \uplus \{ \partial\}$.  
Furthermore, the space $E_\partial$ is also metrizable (see Theorem 5.3 of \cite{Kechris95}) and we will also denote $\Delta$ the corresponding metric.

\begin{definition}
A function $\omega : [0,\infty)\to E_{\partial}$ is called \emph{cadlag} \footnote{cadlag
  stands for the French ``continu \`a droite, limite \`a gauche''} if for every $t \geq
0$, 
\[ \lim_{s > t, s \rightarrow t} \omega(s) = \omega (t) ~\text{ and } ~ \omega(t-) := \lim_{s < t, s \rightarrow t} \omega(s) \text{ exists.} \]

We define a \emph{trajectory} $\omega$ on $E_{\partial}$ to be a
cadlag function $[0,\infty)\to E_{\partial}$ such that
if either $\omega(t-) = \partial$ or
$\omega(t)=\partial$ then $\forall u 
\geq t, \omega(u) = \partial$. 
\end{definition}
As an intuition, a cadlag function $\omega$ is an ``almost continuous'' function with
jumping in a reasonable fashion.  

The intuition behind the additional condition of a trajectory is that once a trajectory
has reached the terminal state $\partial$, it stays in that state and similarly if a
trajectory ``should'' have reached the terminal state $\partial$, then there cannot be a
jump to avoid $\partial$ and once it is in $\partial$, it stays there.  

It is possible to associate to such an FD-semigroup a \emph{canonical FD-process}.
Let $\Omega$ be the set of all trajectories $\omega : [0, \infty) \to
E_\partial$.  
\begin{definition}
\label{def:canonical-FD-process}
The
\emph{canonical FD-process} associated to the FD-semigroup $(\hat{P}_t)_{t \geq 0}$ is
\[(\Omega, \cA, (\cA_t)_{t \geq 0}, (X_t)_{t \geq 0}, (\mathbb{P}^x)_{x \in E_\partial})\]
where
\begin{itemize}
\item the random variable $X_t: \Omega \to E_\partial$ is defined as $X_t(\omega) = \omega (t)$ for every time $t \geq 0$,
\item $\cA = \sigma (X_s ~|~ s \geq 0)$ \footnote{The $\sigma$-algebra $\mathcal{A}$ is the same as the one induced by the Skorohod metric, see theorem 16.6 of \cite{Billingsley99}}, $\cA_t = \sigma (X_s ~|~ 0 \leq s \leq t)$ for every time $t \geq 0$,
\item given any probability measure $\mu$ on $E_\partial$, by the
  Daniell-Kolmogorov theorem, there exists a
  unique probability measure $\mathbb{P}^\mu$ on $(\Omega, \mathcal{A})$
  such that for all
  $n \in \mathbb{N}, 0 \leq t_1 \leq t_2 \leq ...  \leq t_n$ and
  $x_0, x_1, ..., x_n$ in $E_\partial$,
\begin{align*}
   &\mathbb{P}^\mu (X_0 \in dx_0, X_{t_1} \in dx_1, ..., X_{t_n} \in dx_n) \\
  &~~ =
  \mu (dx_0) P_{t_1}^{+\partial}(x_0, dx_1)...P_{t_n -
    t_{n-1}}^{+\partial}(x_{n-1}, dx_n)
  \end{align*}
  where for every $t \geq 0$, $P_t^{+\partial}$ is the Markov kernel obtained by extending $P_t$ to $E_\partial$ by $P_t^{+\partial} (x, \{ \partial \}) = 1 - P_t (x, E)$ and $P_t^{+ \partial} (\partial, \{ \partial \}) = 1$.  The $dx_i$ in this equation should be understood as infinitesimal volumes.  This notation is standard in probability and should be understood by integrating it over measurable state sets $C_i$.
We set $\mathbb{P}^x = \mathbb{P}^{\delta_x}$.
\end{itemize}
\end{definition}
The
distribution $\mathbb{P}^x$ is a measure on the space of
trajectories for a system started at the point~$x$.


   \subsection{Diffusions}

We are now ready to formally define the processes that we will be studying throughout this work: diffusions. 

\begin{definition}
\label{def:honest}
A process described by the semigroup $(\hat{P}_t)_{t \geq 0}$ is \emph{honest} if $\hat{P}_t 1 = 1$ (where $1$ is the constant function on $E$ with value 1), i.e. for every $x \in E$ and every time $t \geq 0$, $P_t(x, E) = 1$
\end{definition}

Worded differently, a process is honest if there is no loss of mass over time. For an honest process, the additional terminal state $\partial$ is never reached.

If furthermore the trajectories of an honest process are continuous, then we may restrict $\Omega$ to the set of continuous functions $\mathbb{R}_{\geq 0} \to E$.  We can equip this set with the compact-open topology $\cT$: it is the topology generated by the subbase formed by the sets $V([0, T], U) = \{ \omega ~|~ \omega([0, T]) \subset U\}$ where $T \geq 0$ and $U$ is an open subset of $E$. Note that the $\sigma$-algebra generated by this topology is $\cA$. 

\begin{definition}
\label{def:discounted-uniform-metric}
Given a 1-bounded pseudometric $m$ on a set $E$ and $0< c < 1$, we define the \emph{discounted uniform pseudometric} $U_c(m)$ on the set of trajectories as follows:
\[ \text{for all } \omega, \omega' \in \Omega, \qquad U_c(m) (\omega, \omega') = \sup_t c^t m(\omega(t), \omega'(t)). \]
\end{definition}
If $m$ is a metric (resp. pseudometric),  then $U_c(m)$ is also a metric (resp. pseudometric). 

The compact-open topology is metrized by the discounted uniform metric $U_c(\Delta)$ for any $0< c< 1$. Furthermore, with this topology, the space $\Omega$ is Polish. The proofs can be found in Appendix \ref{app:CTsystems}.

\begin{definition}
\label{def:diffusion}
A \emph{diffusion} is an honest process satisfying some additional properties:
\begin{itemize}
\item the trajectories in $\Omega$ are continuous,
\item the map $x \mapsto \mathbb{P}^x$ is weakly continuous (wrt the compact-open topology on $\Omega$) meaning that if $y_n \underset{n \rightarrow \infty}{\longrightarrow} y$, then for every bounded and continuous function $f$ on the set of trajectories, $\int f \dee \mathbb{P}^{y_n} \underset{n \rightarrow \infty}{\longrightarrow} \int f \dee \mathbb{P}^{y}$.
\end{itemize}
\end{definition}

In \cite{Rogers00a},  FD-diffusions are also introduced. Our definition is different in the sense that we only define them for honest processes, but we require less conditions. While these conditions may appear restrictive, they are satisfied by classical stochastic processes such as Brownian motion.

\begin{example}
Brownian motion is a stochastic process first introduced to describe the irregular motion of
a particle being buffeted by invisible molecules.  Now its range of
applicability extends far beyond its initial application~\cite{Karatzas12}: it is used to model random noise in
finance, physics, biology, just to mention a few domains.  

We refer the reader to~\cite{Karatzas12} for the complete definition of Brownian motion.  A standard
one-dimensional Brownian motion is a Gaussian process on the real line $\mathbb{R}$ with independent and homogeneous increments. A helpful intuition about Brownian motion is that it is the distributional limit of a
random walk when the time between steps and the distance between states are taken to 0 properly.  
%
\end{example}

   \subsection{Observables}

In previous sections,  we defined FD-processes and diffusions.  In order to bring the processes more in line
with the kind of transition systems that have hitherto been studied in the computer
science literature, we also equip the state space $E$ with an additional continuous function
$obs: E \to [0,1]$. 
One should think of it as the interface between the process and the
user (or an external observer): external observers won't see the exact state in which the process
is at a given time, but they will see the associated \emph{observables}.  What could be a real-life example is the depth at which a diver goes: while the diver does not
know precisely his location underwater, at least his watch is giving him the depth at
which he is.

Note that the condition on the observable is a major difference from the previous work of Chen et al \cite{Chen19a, Chen20} since they used a countable set of atomic propositions $AP$ and $obs$ was a discrete function $E \to 2^{AP}$. 


\section{Defining a pseudometric through functionals}
\label{sec:functionals-metric}

Before diving into details, let us first give a brief overview of this section. We first define a functional $\cF$ on the lattice of 1-bounded pseudometrics based on optimal transport as follows: given a pseudometric $m$ on the state
space and two states $x$ and $y$, we consider the transport cost from $P_t(x)$ to $P_t(y)$ and we define $\cF(m)(x,y)$ to be its supremum over time $t$. 
This functional aims at describing how a metric would be ``deformed'' by the dynamics of the diffusion when the dynamics is described through the kernels $P_t(x)$.
We then define an increasing
sequence of pseudometrics by $\delta_0(x,y) = |obs(x) - obs(y)|$ (i.e.  how different the
states $x$ and $y$ look to an external observer before the evolution of the process starts)
and then iteratively applying the functional $\cF$:  $\delta_{n+1} = \cF(\delta_n)$. 

Our
``behavioural pseudometric'' $\overline{\delta}$ is obtained as $\sup_n \delta_n$.  It further turns
out that the pseudometric $\overline{\delta}$ is a fixpoint of  $\cF$, even though it cannot be defined directly that way. We will see during this section why it is necessary to go through this great length.

In fact, we need to add a discount factor and thus obtain a family of such functionals indexed by the discount factors and a corresponding family of pseudometrics by iteratively applying the functionals.

\subsection{On the set of pseudometrics}

The first part of this work takes place on the lattice of bounded pseudometrics. We are going to introduce the various lattices that we will later be using.

Let $\cM$ be the lattice of 1-bounded pseudometrics on the state space $E$ equipped with the
order $\leq$ defined as: $m_1 \leq m_2$ if and only if for every $(x,y)$,
$m_1(x,y) \leq m_2(x,y)$. 

We further define two sublattices of $\cM$. 

The first one is denoted $\cP$ and is the set of 1-bounded pseudometrics that are lower semi-continuous (wrt the original topology $\cO$ on $E$ generated by the metric $\Delta$ making the space $E$ Polish).

The second sublattice is denoted $\cC$ and is the set of pseudometrics $m \in \cM$ on the state space $E$ such that the topology  generated by $m$ on $E$ is a subtopology of the original topology $\cO$, \emph{i.e.}\  $m$ is a continuous function $E \times E \to [0,1]$.

We have the following inclusion: $\cC \subset \cP \subset \cM$.

\begin{remark}
\label{rem:lattices}
One has to be careful here. The topology $\cO$ on $E$ is generated by the 1-bounded metric $\Delta$, and hence $\Delta$ is in $\cC$. However, we can define many pseudometrics that are not related to $\cO$.  As an example, the discrete pseudometric\footnote{The discrete pseudometric is defined as $m(x, y) = 1$ if $x \neq y$ and $m(x,x) = 0$} on the real line is not related to the usual topology on $\mathbb{R}$. 
We will elaborate on this in Remark \ref{rem:subtlety} once we have the context for our work.
\end{remark}

\subsection{Defining our functional}
\label{sec:def-functional}

Given a discount factor $0 < c< 1$,  we define the functional $\cF_c : \cP \to \cM$ as follows: for every pseudometric $m \in \cP$ and every two states $x,y$,
\[ \cF_c(m)(x,y) = \sup_{t \geq 0} c^t W(m) (P_t(x), P_t(y)).\]
$ \cF_c(m)(x,y) $ compares all the distributions $P_t(x)$ and $P_t(y)$  through transport theory and takes their supremum. 

Recall that using the Kantorovich duality, we have:
\begin{align*}
W(m) (P_t(x), P_t(y)) & = \min_{\gamma \in \Gamma(P_t(x), P_t(y))} \int m ~ \dee \gamma\\
& =  \max_{h \in \cH(m)}  \left(\hat{P}_t h(x) - \hat{P}_t h(y)\right).
\end{align*}

Let us make a few remarks. 
First, we will use the Kantorovich duality throughout this work. It only holds for probability measures which is why we restrict to honest processes. Second, Lemma \ref{lemma:optimal-transport-pseudometric} ensures that $\cF_c(m)$ is indeed in $\cM$, as a supremum of pseudometrics.
Third, the functional $\cF_c$ can be defined for a discount factor $0< c \leq 1$ but we will later require $0<c<1$ (Section \ref{sec:func-restrict-cont}). 

\begin{remark}
\label{rem:subtlety}
We can now further elaborate on Remark \ref{rem:lattices}. As we have mentioned, given a pseudometric $m$ in $\cM$,  $m$ should only be seen as a non-negative cost function that satisfies a few additional constraints such as the triangular inequality. 
The topology generated by $m$ does not need not be related to the topology $\cO$ on $E$.

This is why we restricted the definitions of $\cF_c$ to $\cP$: we need pseudometrics that are lower semi-continuous in order to be able to use transport theory and thus for $\cF_c$ to make sense.

Furthermore, we later on need to iteratively apply $\cF_c$ in search of a fixpoint. While $\cF_c(m)$ is a pseudometric (for $m \in \cP$), there is no reason for it to be in $\cP$. We cannot hastily apply $\cF_c$ again. As we will see in the next section,
we will further restrict to the sublattice $\cC$ of continuous pseudometrics in order to be able to apply our
functional iteratively. 

To sum up, there are two aspects to be cautious about: first, we can only define $\cF_c$ on $\cP$ instead of the whole lattice $\cM$; second, the image under $\cF_c$ of a pseudometric in $\cP$ does not need to be in $\cP$ which means that we cannot iteratively apply $\cF_c$.
This subtlety was not present in ~\cite{Gupta04,Gupta06}.
\end{remark}

As a direct consequence of the definition of $\cF_c$, we have that $\cF_c$  is monotone : if $m_1 \leq m_2$ in $\cP$, then $\cF_c(m_1) \leq \cF_c(m_2)$.

\begin{restatable}{lemma}{lemmaordermFc}
\label{lemma:order-m-Fc}
For every pseudometric $m$ in $\cP$,  discount factor $0 < c \leq 1$ and pair of states~$x,y$,
\[ m(x,y) \leq \cF_c(m)(x,y). \]
\end{restatable}

The proof can be found in the Appendix \ref{app:functionals-metric}

\subsection{When restricted to continuous pseudometrics}
\label{sec:func-restrict-cont}

As noted in Remark \ref{rem:subtlety}, in order to define a sequence of pseudometrics, we first need to make sure that we can indeed apply $\cF_c$ iteratively: given $m \in \cP$,  $\cF_c(m)$ is not necessarily in $\cP$, i.e.  to be lower semi-continuous. However, we will show in Lemma \ref{lemma:Fc-cont-is-cont} that if we restrict to continuous pseudometrics, then we can iteratively apply our functional.

This is where we need that the discount factor $c<1$. The condition that $c< 1$ enables us to maintain continuity by allowing to bound the time interval we consider. Indeed, given $T> 0$, for any time $t \geq T$ and any $x, y \in E$, we know that $c^t W(m) (P_t(x), P_t(y)) \leq~c^T$. 

\begin{lemma}
\label{lemma:Fc-cont-is-cont}
Consider a pseudometric $m \in \cC$.  Then the topology generated by $\cF_c(m)$ is a subtopology of the original topology $\cO$ for any discount factor $0< c< 1$. 
\end{lemma}

\begin{proof}
Since $c$ and $m$ are fixed throughout the proof, we will omit noting them and for instance write $\cF(x,y)$ for $\cF_c(m)(x,y)$ and $W(P_t(x), P_t(y))$ for $W(m)(P_t(x), P_t(y))$. We will also write $\Phi(t,x,y) = c^t W(m)(P_t(x), P_t(y))$, i.e. $\cF(x,y) = \sup_t \Phi(t,x,y)$.

It is enough to show that for a fixed state $x$, the map $y \mapsto \cF(x,y)$ is continuous.

Pick $\epsilon >0$ and a sequence of states $(y_n)_{n \in \mathbb{N}}$ converging to $y$.  We want to show that there exists $M$ such that for all $n \geq M$,
\begin{align}
\label{eq:limitF}
|\cF(x, y) - \cF(x, y_n)| & \leq \epsilon.
\end{align}

Pick $t$ such that $\cF(x,y) = \sup_s \Phi(s, x,y) \leq \Phi(t, x, y) + \epsilon / 4$, i.e.
\begin{align}
\label{eq:numbera}
| \Phi(t, x, y)  - \cF(x,y)| \leq  \epsilon / 4.
\end{align}

Note that by the second condition of Definition \ref{def:diffusion}, $P_t(y_n)$ converges weakly to $P_t(y)$ and hence we can apply Theorem 5.20 of \cite{Villani08} and get:
\[ \lim_{n \rightarrow \infty} W(P_t(x), P_t(y_n)) = W(P_t(x), P_t(y)). \]
This means that there exists $N'$ such that for all $n \geq N'$, 
\[ |W(P_t(x), P_t(y_n)) - W(P_t(x), P_t(y))| \leq \epsilon/4. \]
This further implies that for all $n \geq N'$, 
\begin{align}
\label{eq:numberb}
|\Phi(t, x, y_n) - \Phi(t, x, y)|
&  \leq c^t \epsilon/4  \leq \epsilon / 4.
\end{align}

In order to show (\ref{eq:limitF}), it is enough to show that there exists $N$ such that for every $n \geq N$,
\begin{align}
\label{eq:limitF-intermediary}
|\Phi(t, x, y_n) - \cF(x, y_n)| & \leq \epsilon / 2.
\end{align}
Indeed, in that case, $\forall n \geq \max\{N, N'\}$,  using Equations (\ref{eq:numbera}) and (\ref{eq:numberb}),
\begin{align*}
&|\cF(x, y) - \cF(x, y_n)|\\
& \leq |\cF(x, y) - \Phi(t, x, y)| + |\Phi(t, x, y) - \Phi(t, x, y_n)| \\
& \qquad + |\Phi(t, x, y_n) - \cF(x, y_n)| \\
& \leq \epsilon / 4 + \epsilon / 4 + \epsilon / 2 = \epsilon.
\end{align*}

So let us show (\ref{eq:limitF-intermediary}). Assume it is not the case: for all $N$, there exists $n \geq N$ such that $ |\Phi(t, x, y_n) - \cF(x, y_n)| > \epsilon / 2$, i.e.
\[ \Phi(t, x, y_n) + \epsilon / 2 <  \cF(x, y_n).  \]
Define the sequence $(N_k)_{k \in \mathbb{N}}$ by: $N_{-1} = -1$ and if $N_k$ is defined, define $N_{k+1}$ to be the smallest $n \geq N_k + 1$ such that $\Phi(t, x, y_n) + \epsilon / 2 <  \cF(x, y_n)$.
In particular for every $k \in \mathbb{N}$,  $\cF(x, y_{N_k})  > \epsilon / 2$. There exists $T$ such that for every $s \geq T$, $c^s < \epsilon/2$.  We thus have that
\[ \forall k \in \mathbb{N} \quad \cF(x, y_{N_k}) = \sup_{ 0 \leq s \leq T} \Phi(s, x, y_{N_k}).\]
Therefore for every $k \in \mathbb{N}$, there exists $s_k \in [0, T]$ such that
\begin{align}
\label{eq:numberc}
\cF(x, y_{N_k}) \leq \Phi(s_k, x, y_{N_k}) + \epsilon / 8.
\end{align}
We get a sequence $(s_k)_{k \in \mathbb{N}} \subset [0,T]$, and there is thus a subsequence $(t_k)_{k \in \mathbb{N}}$ converging to some $t' \in [0, T]$. There is a corresponding subsequence $(z_k)_{k \in \mathbb{N}}$ of the original sequence $(y_{N_k})_{k \in \mathbb{N}}$. Since $\lim_{n \rightarrow \infty} y_n = y$ , then $\lim_{k \rightarrow \infty} z_k = y$. 

We constructed the sequence $(N_k)_{k \in \mathbb{N}}$ such that  $\Phi(t, x, y_{N_k}) + \epsilon / 2 <  \cF(x, y_{N_k})$. Hence by Equation (\ref{eq:numberc}),
\[  \Phi(t, x, z_k) + \epsilon / 2 < \cF (x, y_{N_k}) \leq \Phi(t_k, x, z_k) + \epsilon / 8,  \]
which means that by taking the limit $k \to \infty$,
\begin{align}
\label{eq:contr1}
 \Phi(t, x, y) + \epsilon / 2 &\leq \Phi(t', x, y) + \epsilon / 8.
\end{align}

At the very start of this proof, we had picked $t$ such that $\cF(x,y) = \sup_s \Phi(s, x,y) \leq \Phi(t, x, y) + \epsilon / 4$ which means that
\begin{align}
\label{eq:contr2}
 \Phi(t', x, y) &\leq \Phi(t, x,y) + \epsilon / 4 .
\end{align}
Equations (\ref{eq:contr1}) and (\ref{eq:contr2}) are incompatible which concludes the proof.
\end{proof}

\subsection{Defining our family of pseudometrics}
\label{sec:def-pseudometrics-F}

We are now finally able to iteratively apply our functional $\cF_c$ on continuous pseudometrics.

Since $obs$ is a continuous function $E \to \mathbb{R}_{\geq 0}$ and from Lemma \ref{lemma:Fc-cont-is-cont}, we can define the sequence of pseudometrics in $\cC$ for each $0 < c< 1$:
\begin{align*}
\delta^c_0(x,y) & = |obs(x) - obs(y)|,\\
\delta^c_{n+1} &= \cF_c(\delta_n^c).
\end{align*}
From Lemma \ref{lemma:order-m-Fc}, for every two states $x$ and $y$, $\delta^c_{n+1}(x,y) \geq \delta^c_n(x,y)$.  Define the pseudometric $\overline{\delta}^c = \sup_n \delta^c_n$ (which is also a limit since the sequence is non-decreasing). 

As a direct consequence of Lemma \ref{lemma:sup-lower-semicontinuous},  the pseudometric $\overline{\delta}^c$ is lower semi-continuous and is thus in the lattice $\cP$ for any $0 < c < 1$.

\begin{remark}
\label{rem:tarski}
Note that the lattice $\cC$ is not complete which means that, although the metrics $\delta^c_n$ all belong to $\cC$,   $\overline{\delta}^c$ does not need to be in $\cC$.  For that reason, we cannot use the Knaster-Tarski theorem in this work.
\end{remark}

\subsection{Fixpoints}
\label{sec:metrics-fixpoints}

Even though we are not able to define the metric  $\overline{\delta}^c$ as a fixpoint directly, it is actually a fixpoint.

\begin{theorem}
\label{thm:m-fixpointFc}
The pseudometric $\overline{\delta}^c$ is a fixpoint for $\cF_c$.
\end{theorem}

\begin{proof}
We will omit writing $c$ as an index for the pseudometrics $\delta_n$ and $\overline{\delta}$ throughout this proof. Fix two states $x,y$ and a time $t$.

The space of finite measures on $E \times E$ is a linear topological space.
Using Lemma \ref{lemma:couplings-compact}, we know that the set of couplings $\Gamma (P_t(x), P_t(y))$ is a compact subset.  Besides, it is also convex.

The space of bounded pseudometrics on $E$ is also a linear topological space.  We have defined a sequence $\delta_0, \delta_1, ...$ on that space.  We define $Y$ to be the set of linear combinations of pseudometrics $\sum_{n \in \mathbb{N}} a_n \delta_n$ such that for every $n$, $a_n \geq 0$ (and finitely many are non-zero) and $\sum_{n \in \mathbb{N}} a_n = 1$.  This set $Y$ is convex.

Define the function
\begin{align*}
\Xi : \Gamma (P_t(x), P_t(y)) \times Y & \to [0,1]\\
(\gamma, m) & \mapsto\int m ~ \dee \gamma.
\end{align*}

For $\gamma \in \Gamma (P_t(x), P_t(y))$, the map $\Xi(\gamma, \cdot)$ is continuous by dominated convergence theorem and it is monotone and hence quasiconcave.  For a given $m \in Y$, by definition of the Lebesgue integral, $\Xi(\cdot, m)$ is continuous and linear.

We can therefore apply Sion's minimax theorem:
\begin{align}
\label{eq:minimax-F}
\min_{\gamma \in \Gamma (P_t(x), P_t(y))} \sup_{m \in Y} \int m ~ \dee \gamma =  \sup_{m \in Y} \min_{\gamma \in \Gamma (P_t(x), P_t(y))}  \int m ~ \dee \gamma.
\end{align}

Now note that for an arbitrary functional $\Psi: Y \to \mathbb{R}$ such that for any two pseudometrics $m$ and $m'$,  $m \leq m' \Rightarrow \Psi(m) \leq \Psi(m')$, 
\begin{align}
\label{eq:numberd}
\sup_{m \in Y} \Psi(m) = \sup_{n \in \mathbb{N}} \Psi (\delta_n).
\end{align}

Indeed since $\delta_n \in Y$ for every $n$,  we have that $\sup_{m \in Y} \Psi(m) \geq \sup_{n \in \mathbb{N}} \Psi (\delta_n)$.  For the other direction, note that for every $m \in Y$, there exists $n$ such that for all $k \geq n$, $a_k = 0$.  Thus $m \leq \delta_n$ and therefore $\Psi(m) \leq \Psi(\delta_n)$ which is enough to prove the other direction.

The right-hand side of Equation (\ref{eq:minimax-F}) is
\begin{align*}
& \sup_{m \in Y} \min_{\gamma \in \Gamma (P_t(x), P_t(y))}  \int m ~ \dee \gamma = \sup_{m \in Y} W(m)(P_t(x), P_t(y))\\
& = \sup_{n \in \mathbb{N}} W(\delta_n)(P_t(x), P_t(y))\quad \text{(previous result in Equation (\ref{eq:numberd})).}
\end{align*}

The left-hand side of Equation (\ref{eq:minimax-F}) is
\begin{align*}
& \min_{\gamma \in  \Gamma (P_t(x), P_t(y))} \sup_{m \in Y} \int m ~ \dee \gamma\\
& = \min_{\gamma \in  \Gamma (P_t(x), P_t(y))} \sup_{n \in \mathbb{N}} \int \delta_n ~ \dee \gamma \quad \text{(previous result in Equation (\ref{eq:numberd}))}\\
& = \min_{\gamma \in  \Gamma (P_t(x), P_t(y))}  \int \sup_{n \in \mathbb{N}} \delta_n ~ \dee \gamma \quad \text{(dominated convergence theorem)}\\
& = \min_{\gamma \in  \Gamma (P_t(x), P_t(y))}  \int \overline{\delta} ~ \dee \gamma \\
& = W(\overline{\delta})(P_t(x), P_t(y)).
\end{align*}

We thus have that
\begin{align*}
& \cF_c(\overline{\delta}) (x, y) = \sup_{t \geq 0} c^t W(\overline{\delta})(P_t(x), P_t(y))\\
& = \sup_{t \geq 0} c^t \sup_{n \in \mathbb{N}} W(\delta_n)(P_t(x), P_t(y)) \quad \text{using Sion's minimax theorem}\\
& = \sup_{n \in \mathbb{N}}  \sup_{t \geq 0} c^t W(\delta_n)(P_t(x), P_t(y)) = \sup_{n \in \mathbb{N}}  \cF_c(\delta_n)(x,y)\\
& =  \sup_{n \in \mathbb{N}} \delta_{n+1} (x,y) \quad \text{by definition of } (\delta_n)_{n \in \mathbb{N}}\\
& = \overline{\delta}(x,y) \quad \text{by definition of } \overline{\delta}.
\end{align*}
\end{proof}

\begin{lemma}
\label{lemma:comparison-fixpointFc}
Consider a discount factor $0 < c < 1$ and a pseudometric $m$ in $\cP$ such that
\begin{enumerate}
\item $m$ is a fixpoint for $\cF_c$,
\item for every to states $x$ and $y$, $m(x,y) \geq |obs(x) - obs(y)|$,
\end{enumerate}
then $m \geq \overline{\delta}^c$.
\end{lemma}

\begin{proof}
We show by induction on $n$ that $m \geq \delta_n^c$.
First, we have that $m \geq \delta_0^c$. 
Then assume it is true for $n$, i.e. $m \geq \delta_n^c$ and let us show that it also holds for $n+1$.  Since the functional $\cF_c$ is monotone:
\[ m = \cF_c(m) \geq \cF_c(\delta_n^c) = \delta_{n+1}^c. \]

Since for every $n$, $m \geq \delta_n^c$,  $m \geq \sup_n \delta_n^c = \overline{\delta}^c$.
\end{proof}

What Lemma \ref{lemma:comparison-fixpointFc} and Theorem \ref{thm:m-fixpointFc} imply is the following characte\-rization of $\overline{\delta}^c$.
\begin{thm}
The pseudometric  $\overline{\delta}^c$ is the least fixpoint of $\cF_c$ that is greater than the pseudometric $(x, y) \mapsto |obs(x) - obs(y)|$.
\end{thm}

\section{Corresponding real-valued logic}
\label{sec:logics}

In this section, we introduce a real-valued logic. This real-valued logic should be thought of as tests performed on the diffusion process, for instance ``what is the expected value of $obs$ after letting the process evolve for time $t$?'' and it generates a pseudometric on the state space by looking at how different the process performs on those tests starting from different positions.

In the previous section, we have defined a pseudometric that is a fixpoint of a functional. In this section, we are going to show that that fixpoint pseudometric coincides with the pseudometric generated by the real-valued logic.

\subsection{Definition of the real-valued logic}

In order to define our real-valued logic, we first introduce a grammar of expressions.  This grammar of expressions is then interpreted as a family of functions for a given diffusion. 

\paragraph*{Definition of the  logics:}
The logic is defined inductively and is denoted $\grammar$:
\begin{align*}
f \in \grammar & := q~|~ obs ~|~ \min\{ f_1, f_2 \} ~|~ 1-f ~|~ f \ominus q ~|~ \langle t \rangle f
\end{align*}
for all $f_1, f_2, f \in \grammar$, $q \in [0,1] \cap \mathbb{Q}$ and $t \in \mathbb{Q}_{ \geq 0}$.

This logic closely resembles the ones introduced for discrete-time systems by Desharnais et al.~\cite{Desharnais99b, Desharnais04} and by van Breugel et al.~\cite{vanBreugel05a}. The key difference is the term $\langle t \rangle f$ which deals with continuous-time.

\paragraph*{Interpretation of the logics: }
We fix a discount factor $0 < c< 1$. The expressions in $\grammar$ are interpreted inductively as functions $E \to [0,1]$ as follows for a state $x \in E$:
\begin{align*}
q(x) & = q, \\
obs(x) & = obs(x), \\
(\min\{ f_1, f_2 \})(x) & = \min\{ f_1(x),  f_2(x) \}, \\
(1-f)(x) & = 1 - f(x), \\
(f \ominus q)(x) & = \max \{ 0, f(x) - q \}, \\
(\langle t \rangle f)(x) &= c^t ~ \int f(y) ~ P_t(x, \dee y) = c^t~ \left(\hat{P}_t f\right)(x).
\end{align*}
Whenever we want to emphasize the fact that the expressions are interpreted for a discount factor $0<c<1$, we will write $\grammar_c$.

\begin{remark}
Let us clarify what the difference is between an expression in $\grammar$ and its
interpretation. Expressions can be thought of as the
notation $+, ^2, \times$ etc. They don't carry much meaning by
themselves but one can then interpret them for a given set: $\mathbb{R},  C_0(E)$
(continuous functions $E \to \mathbb{R}$ that vanish at infinity) for instance. Combining
notations, one can write expressions that can then be interpreted
on a given set.
\end{remark}

\paragraph*{Additional useful expressions}
From $\grammar$, we can define the following two expressions
\begin{align*}
\max\{ f_1, f_2 \} & = 1 - \min \{1-f_1, 1-f_2 \},\\
f \oplus q &  = 1 - ((1 - f) \ominus q),
\end{align*}
which are interpreted as functions $E \to [0,1]$ as:
\begin{align*}
 (\max\{ f_1, f_2 \})(x) & = \max\{ f_1(x),  f_2(x) \}, \\
 (f \oplus q)(x)  & = \min \{ 1, f(x) + q \}.
\end{align*}

\subsection{Definition of the pseudometric}

The pseudometric we derive from the logic $\grammar$ corresponds to how different the test results are when the process starts from $x$  compared to the case when it starts from $y$.

Given a fixed discount factor $0< c<1$, we can define the pseudometric $\lambda^c$:
\[ \lambda^c (x,y) = \sup_{f \in \grammar_c} (f (x) - f(y)) = \sup_{f \in \grammar_c} |f(x) - f(y)|.  \]
The latter equality holds since for every $f \in \grammar_c$, $\grammar_c$ also contains $1-f$.
%

\subsection{Comparison to the fixpoint metric}
This real-valued logic $\grammar_c$ is especially interesting as the corresponding pseudometric $\lambda^c$ matches the fixpoint pseudometric $\overline{\delta}^c$ for the functional $\cF_c$ that we defined in Section \ref{sec:def-pseudometrics-F}.
In order to show that $\lambda^c = \overline{\delta}^c$, we establish the inequalities in both directions.

\begin{restatable}{lemma}{lemmaalphaf}
\label{lemma:alpha-f}
For every $f$ in $\grammar_c$, there exists $n$ such that for every $x,y$, $ f(x) - f(y) \leq \delta^c_n(x,y)$.
\end{restatable}

This proof is done by induction on the structure of $f$.  A full version can be found in Appendix \ref{app:sec-logics}.

Since $\lambda^c (x,y) = \sup_{f \in \grammar_c} ( f(x) - f(y) )$,  we get the following corollary as an immediate consequence of Lemma \ref{lemma:alpha-f}.
\begin{cor}
\label{cor:functional-F-geq-logic}
For every discount factor $0< c< 1$,  $\overline{\delta}^c \geq \lambda^c$.
\end{cor}

We now aim at proving the reverse inequality.
We will use the following result (Lemma A.7.2) from \cite{Ash72} which is also used in~\cite{vanBreugel05a} for discrete-time processes.
\begin{lemma}
\label{lemma:magie}
Let $X$ be a compact Hausdorff space.  Let $A$ be a subset of the set of continuous functions $X \to \mathbb{R}$ such that if $f, g \in A$, then $\max\{ f,g\}$ and $\min \{f,g\}$ are also in $A$.  If a function $h$ can be approximated at each pair of points by functions in $A$, then $h$ can be approximated by functions in $A$.
\end{lemma}

In order to use this lemma, we need the following lemmas:

\begin{lemma}
\label{lemma:grammar-continuous}
For every $f \in \grammar_c$, the function $x \mapsto f(x)$ is continuous.
\end{lemma}

\begin{proof}
This is done by induction on the structure of $\Lambda_c$. The only case that is not straightforward is when $f =  \langle t \rangle g$. By induction hypothesis, $g$ is continuous.
Since the map $x \mapsto P_t(x)$ is continuous (onto the weak topology), $f = c^t \hat{P}_t g$ is continuous.
\end{proof}

\begin{restatable}{lemma}{lemmaprepmagieF}
\label{lemma:prep-magie-F}
Consider a continuous function $h: E \to [0,1]$ and two states $z, z'$ such that there exists $f$ in the logic $\grammar$ such that $|h(z) - h(z')| \leq |f(z) - f(z')|$. Then for every $\delta > 0$, there exists $g \in \grammar$  such that $|h(z) - g(z)| \leq 2 \delta$ and $|h(z') - g(z')| \leq 2 \delta$.
\end{restatable}

\begin{proof}
WLOG $h(z) \geq h(z')$ and $f(z) \geq f(z')$ (otherwise consider $1-f$ instead of $f$).
Pick $p,q,r \in \mathbb{Q} \cap [0,1]$ such that
\begin{align*}
p & \in [f(z') - \delta, f(z')],\\
q & \in [h(z) - h(z') - \delta, h(z) - h(z')],\\
r & \in [h(z'), h(z') + \delta].
\end{align*}
Define $g = ( \min \{ f \ominus p, q \} ) \oplus r$. The computations showing that the result follows can be found in the Appendix \ref{app:sec-logics}.
\end{proof}

\begin{cor}
\label{cor:prep-magie}
Consider a continuous function $h: E \to [0,1]$ such that for any two states $z, z'$ there exists $f$ in the logic $\grammar$ such that $|h(z) - h(z')| \leq |f(z) - f(z')|$. Then for every compact set $K$ in $E$, there exists a sequence $(g_n)_{n \in \mathbb{N}}$ in $\Lambda$ that approximates $h$ on~$K$.
\end{cor}

\begin{proof}
We have proven in Lemma \ref{lemma:prep-magie-F} that such a function $h$ can be approximated at pairs of states by functions in $\grammar$. Now recall that all the functions in $\grammar$ are continuous (Lemma \ref{lemma:grammar-continuous}).

We can thus apply Lemma \ref{lemma:magie} on the compact set $K$, and we get that the function $h$ can be approximated by functions in $\grammar$.
\end{proof}

\begin{theorem}
\label{thm:lambda-fixpoint-F}
The pseudometric $\lambda^c$ is a fixpoint of $\cF_c$.
\end{theorem}

\begin{proof}
We will omit writing $c$ as an index in this proof.  We already have that $\lambda \leq \cF(\lambda)$ (cf Lemma \ref{lemma:order-m-Fc}), so we only need to prove the reverse direction.

There are countably many expressions in $\grammar$ so we can number them: $\grammar = \{ f_0, f_1, ...\}$. Write $m_k (x,y) = \max_{j \leq k} |f_k(x) - f_k(y)|$. 
Since all the $f_j$ are continuous (Lemma \ref{lemma:grammar-continuous}), the map $m_k$ is also continuous. Furthermore, $m_k \leq m_{k+1}$ and for every two states $x$ and $y$, $\lim_{k \rightarrow \infty} m_k(x,y) = \lambda (x,y)$.

Using Lemma \ref{lemma:wasserstein-limit}, we know that for every states $x,y$ and time $t$, 
\[\sup_k W(m_k)(P_t(x), P_t(y)) = W(\lambda)(P_t(x), P_t(y)).\]
This implies that
\begin{align*}
\cF (\lambda) (x, y) 
&= \sup_{t \geq 0} c^t W(\lambda)(P_t(x), P_t(y))\\
& = \sup_{k} \sup_{t \geq 0} c^t W(m_k)(P_t(x), P_t(y)).
\end{align*}

It is therefore enough to show that for every $k$, every time $t \in \mathbb{Q}_{\geq 0}$ and every pair of states $x,y$, $c^t W(m_k)(P_t(x), P_t(y)) \leq \lambda (x,y)$.
There exists $h \in \cH(m_k)$ such that $W(m_k)(P_t(x), P_t(y)) = \left| \int h~ \dee P_t(x) - \int h~ \dee P_t(y)\right| .$

Since $P_t(x)$ and $P_t(y)$ are tight, for every $\epsilon > 0$, there exists a compact set $K \subset E$ such that $P_t(x, E \setminus K) \leq \epsilon / 4$ and $P_t(y, E \setminus K) \leq \epsilon / 4$.

By Corollary \ref{cor:prep-magie}, there exists $(g_n)_{n \in \mathbb{N}}$ in $\Lambda$ that pointwise converge to $h$ on $K$.  In particular, for $n$ large enough,
\begin{align*}
\left| \int_K g_n ~ \dee P_t(x) - \int_K h ~ \dee P_t(x) \right| \leq \epsilon /4,
\end{align*}
and similarly for $P_t(y)$. We get that for $n$ large enough.
\begin{align*}
& \left| \int_E g_n ~ \dee P_t(x) - \int_E h ~ \dee P_t(x) \right| \\
& \leq \left| \int_K g_n ~ \dee P_t(x) - \int_K h ~ \dee P_t(x) \right| + \int_{E \setminus K} |g_n- h| ~ \dee P_t(x) \\
& \leq \epsilon / 2,
\end{align*}
and similarly for $P_t(y)$. We can thus conclude that
\begin{align*}
& c^t W(m_k) (P_t(x), P_t(y)) = c^t \left| \int_E h ~ \dee P_t(x) - \int_E h ~ \dee P_t(y) \right|\\
& \leq c^t \left| \int_E g_n ~ \dee P_t(x) - \int_E h ~ \dee P_t(x) \right| + c^t \left| \int_E g_n ~ \dee P_t(y) - \int_E h ~ \dee P_t(y) \right|\\
& \qquad + \left| (\langle t \rangle g_n)(x) - ( \langle t \rangle g_n)(y)\right|\\
& \leq \epsilon + \lambda(x, y).
\end{align*}
Since $\epsilon$ is arbitrary,  $c^t W(m_k) (P_t(x), P_t(y)) \leq \lambda (x, y)$.

\end{proof}



As a consequence of Theorem \ref{thm:lambda-fixpoint-F} and using Lemma \ref{lemma:comparison-fixpointFc}, we get that $\overline{\delta}^c \leq \lambda^c$. Then with Corollary \ref{cor:functional-F-geq-logic}, we finally get:
\begin{theorem}
\label{thm:metric-equal-logic-functional-F}
The two pseudometrics are equal: $\overline{\delta}^c = \lambda^c$.
\end{theorem}

\section{Another approach through trajectories}
\label{sec:diffusions}

In previous sections, the dynamics of a diffusion was described by the Markov kernels $P_t(x)$. However, as we have seen in Section \ref{sec:CTsystems},  the dynamics of the system can also be described by the probabilities $\mathbb{P}^x$ on the space of trajectories. This leads to another pseudometric that can similarly be defined both through a functional and through a real-valued logic. The general idea is similar to that in Sections \ref{sec:functionals-metric} and \ref{sec:logics} but the implementation requires more work as we no longer deal only with states but also with trajectories.

\subsection{Through functionals}

Similarly as in Section \ref{sec:functionals-metric}, we will define a functional $\cG$ on the lattice $\cP$ of lower semi-continuous 1-bounded pseudometrics on $E$, and, by applying $\cG$ iteratively on the lattice of continuous pseudometrics, seek the fixpoint of $\cG$.  We also need a discount factor $0 < c < 1$ in this part.

%
%


\subsubsection{Some more results on our lattices}

Recall that we have defined three lattices $\cC \subset \cP \subset \cM$ where $\cM$ is the lattice of 1-bounded pseudometrics on $E$, $\cP$ is the sublattice of lower semi-continuous pseudometrics and $\cC$ of continuous pseudometrics. 

In Section \ref{sec:functionals-metric}, we restricted the pseudometrics, which were used as cost functions, to  the lattice $\cP$ or to the lattice $\cC$. 
Here,   since we are working with trajectories, we need to define a cost function on $\Omega$ the space of trajectories. This is done through a discounted uniform pseudometric (see Definition \ref{def:discounted-uniform-metric}), so we need that the regularity conditions on a pseudometric $m$ are also passed on to $U_c(m)$. Recall that $\Omega$ is equipped with the compact-open topology which is metrized by $U_c(\Delta)$.

\begin{lemma}
\label{lemma:m-lsc-Uc(m)-lsc}
Let $m$ be a pseudometric in $\cP$.
Then, $U_c(m): \Omega\times\Omega\to\mathbb{R}$ is lower semi-continuous for every $0 < c< 1$.
\end{lemma}

\begin{proof}
For $r \in [0,1]$, let us consider the set $U_c(m)^{-1}((r, 1]) = \{ (\omega, \omega') ~|~ U_c(m)(\omega, \omega') > r \}$. We intend to show that this set is open in $\Omega \times \Omega$. 

Consider $(\omega, \omega') \in U_c(m)^{-1}((r, 1])$.  This means that 
\[\sup_t c^t m (\omega(t), \omega'(t)) > r\]
 which implies that there exists $t$ such that $m (\omega(t), \omega'(t))> r c^{-t}$.

Define the set $O = \{ (x,y) ~|~ m(x,y) >rc^{-t} \} \subset E \times E$. This set $O$ is open in $E \times E$ since the function $m$ is lower semi-continuous.  

The pair $(\omega(t), \omega'(t))$ is in the open set $O$ which means that there exists $q > 0$ \footnote{Note that this $q$ also depends on how the topology on $E \times E$ is metrized, for instance $p$-distance or max distance} such that $B \times B' \subset O$ where the two sets $B$ and $B'$ are defined as open balls in $E$:
\[ B = \{ z ~|~ \Delta (z, \omega(t)) < qc^{-t} \} \quad \text{and} \quad B'= \{ z ~|~ \Delta (z, \omega'(t)) < qc^{-t} \}. \]

Define the two sets $A$ and $A'$ as
\[ A = \{ \theta ~|~ U_c(\Delta)(\omega, \theta) < q\} \qquad \text{and} \qquad A' = \{ \theta ~|~ U_c(\Delta)(\omega', \theta) < q\}. \]
These are open balls in $\Omega$ and thus $A \times A'$ is open in $\Omega \times \Omega$. 

The set $A \times A'$ is included in $U_c(m)^{-1}((r, 1])$.  Indeed, consider $\theta \in A$ and $\theta' \in A'$. This means that $U_c(\Delta)(\omega, \theta) < q$, which implies that $c^t \Delta (\omega(t), \theta(t)) < q$, i.e. $\theta(t) \in B$. Similarly, $\theta'(t) \in B'$, which means that $(\theta(t), \theta'(t)) \in B \times B' \subset O$, which means that $c^t m(\theta(t), \theta'(t)) >r $.

Furthermore, $\omega \in A$ and $\omega' \in A'$ which means that $A \times A'$ is an open neighbourhood of $(\omega, \omega')$ in $U_c(m)^{-1}((r, 1])$ which concludes the proof that $U_c(m)^{-1}((r, 1])$ is open. 
\end{proof}

\begin{lemma}
\label{lemma:m-cont-Uc(m)-cont}
If $m \in \cC$, then $U_c(m)$ is continuous as a function $\Omega \times \Omega \to \mathbb{R}$ for every $0 < c< 1$.
\end{lemma}

\begin{proof}
It is enough to show that if $(\omega_n)_{n \in \mathbb{N}}$ is a sequence of trajectories that converges to $\omega$ (in the topology generated by $U_c(\Delta)$), then $\lim_{n \to \infty} U_c(m)(\omega, \omega_n) = 0$.
Assume it is not the case: there exists $\epsilon > 0$ and an increasing sequence of natural numbers $(n_k)_{k \in \mathbb{N}}$ such that $U_c(m) (\omega, \omega_{n_k}) \geq 3 \epsilon$. This means that for every $k$, there exists a time $t_k$ such that $c^{t_k} m (\omega(t_k), \omega_{n_k}(t_k)) \geq 2 \epsilon$. Since $m$ is 1-bounded, those times $t_k$ are bounded by $\ln (2 \epsilon)/ \ln c$.

For that reason we may assume that  the sequence $(t_k)_{k \in \mathbb{N}}$ converges to some time $t$ (otherwise there is a subsequence that converges, so we can consider that subsequence instead and the corresponding trajectories $\omega_{n_k}$). We have the following inequalities for every $k$:
\begin{align*}
2 \epsilon
& \leq c^{t_k} m(\omega (t_k), \omega_{n_k}(t_k))\\
& \leq c^{t_k} m(\omega (t_k), \omega(t)) + c^{t_k} m(\omega (t), \omega_{n_k}(t_k))\\
& \leq m(\omega (t_k), \omega(t)) +  m(\omega (t), \omega_{n_k}(t_k))
\end{align*}
The trajectories and $m$ are continuous,  so
$\lim_{k \rightarrow \infty} m(\omega (t_k), \omega(t)) =~0$, i.e. there exists $K$ such that for all $k \geq K$,  $m(\omega (t_k), \omega(t)) < \epsilon$. We get that for all $k \geq K$,  $\epsilon < m(\omega(t), \omega_{n_k}(t_k))$.

Now recall that $(\omega_n)_{n \in \mathbb{N}}$ converges to $\omega$, i.e.  $\lim_{n \rightarrow \infty} U_c(\Delta)(\omega_n, \omega) = 0$. Using the triangular inequality, we have that
\begin{align*}
\Delta (\omega_{n_k}(t_k), \omega(t))
& \leq \Delta (\omega_{n_k}(t_k), \omega(t_k)) + \Delta (\omega(t_k), \omega(t))\\
&\leq c^{t_k} U_c(\Delta) (\omega_{n_k}, \omega)  + \Delta (\omega(t_k), \omega(t))\\
&\leq \frac{1}{2\epsilon} U_c(\Delta) (\omega_{n_k}, \omega)  + \Delta (\omega(t_k), \omega(t)).
\end{align*}
The second term of the right-hand side converges to 0 as $k \rightarrow \infty$ since $\omega$ is continuous. Thus $\lim_{k  \rightarrow \infty} \Delta (\omega_{n_k}(t_k), \omega(t)) = 0$. 
This implies that  $\lim_{k  \rightarrow \infty} m (\omega_{n_k}(t_k), \omega(t)) = 0$ which directly contradicts that $\epsilon < m(\omega(t), \omega_{n_k}(t_k))$ for every $k \geq K$.
\end{proof}

\subsubsection{The family of functionals}

The (family of) functional(s) $\cF_c$ that we defined in Section \ref{sec:functionals-metric} compared the distributions $P_t(x)$ on the state space. 
The new functional that we will define here compares  the probability distributions $\mathbb{P}^x$ on the space of trajectories through transport theory.  
Again, this new functional also involves a discount factor $0< c< 1$,  and it is denoted $\cG_c$. 

Given $0 < c<1$,  we define the functional $\cG_c: \cP \to \cM$ as follows: for every pseudometric $m \in \cP$ and every two states $x,y$,
\[ \cG_c(m)(x,y) = W(U_c(m))(\mathbb{P}^x, \mathbb{P}^y).\]
This is well-defined as we have proven in Lemma \ref{lemma:m-lsc-Uc(m)-lsc} that if $m$ is lower semi-continuous, then so is $U_c(m)$.

Using the Kantorovich duality, we have:
\begin{align*}
\cG_c(m)(x,y)  & =  \min_{\gamma \in \Gamma(\mathbb{P}^x, \mathbb{P}^y)} \int U_c(m) ~ \dee \gamma\\
& =  \sup_{h \in \cH(U_c(m))} \left(\int h ~ \dee  \mathbb{P}^x - \int h ~ \dee  \mathbb{P}^y\right).
\end{align*}
Lemma \ref{lemma:optimal-transport-pseudometric} ensures that $\cG_c(m)$ is indeed in $\cM$.

The remarks (\ref{rem:subtlety} and \ref{rem:tarski}) that we made on $\cF_c$ still apply here. In particular, we will need to restrict to the lattice $\cC$ in order for $\cG_c$ to be applied iteratively.  For the same reason as before (i.e. the lattice $\cC$ is not complete), we cannot apply the Knaster-Tarski theorem so we will have to go through similar steps as for $\cF_c$.

\subsubsection{Some ordering results on our functionals}

As a direct consequence from the definition of $\cG_c$, we have that  $\cG_c$ is monotone (for any $0 < c < 1$) : if $m_1 \leq m_2$ in $\cP$, then $\cG_c(m_1) \leq \cG_c(m_2)$.\\
The proof of the two following results can be found in Appendix~\ref{app:diffusions}.
\begin{restatable}{lemma}{lemmaordermFcGc}
\label{lemma:order-m-Fc-Gc}
For every pseudometric $m$ in $\cP$, every discount factor $0 < c < 1$ and every pair of states~$x,y$,
\[ m(x,y) \leq \cF_c(m)(x,y) \leq \cG_c(m)(x,y). \]
\end{restatable}

\begin{restatable}{cor}{corfixpointGimpliesF}
\label{cor:fixpoint-G-implies-F}
A fixpoint for $\cG_c$ is also a fixpoint for $\cF_c$.
\end{restatable}

\subsubsection{When restricted to continuous pseudometrics}

As in Section \ref{sec:functionals-metric}, one does not have that $\cG_c(m) \in \cP$ when $m \in \cP$ and we have to go through the lattice of continuous pseudometrics $\cC$.

\begin{lemma}
\label{lemma:Gc-cont-is-cont}
Consider a pseudometric $m \in \cC$.  Then the topology on $E$ generated by $\cG_c(m)$ is a subtopology of the original topology $\cO$ for any $0 < c< 1$.
\end{lemma}

\begin{proof}
Using Lemma \ref{lemma:m-cont-Uc(m)-cont}, we know that $U_c(m)$ is continuous wrt the topology on $\Omega \times \Omega$ generated by $U_c(\Delta)$.

In order to show that $\cG_c(m)$ generates a topology on $E$ which is a subtopology of $\cO$, it is enough to show that for a fixed state $x \in E$, the map $y \mapsto \cG_c(m)(x,y)$ is continuous on $E$. 

Pick a sequence $(y_n)_{n \in \mathbb{N}}$ in $E$ that converges to some $y \in E$ in the topology $\cO$. Recall that $\mathbb{P}^{y_n}$ converges weakly to $\mathbb{P}^y$. Then,
\begin{align*}
 \lim_{n \rightarrow \infty} \cG_c(m)(x, y_n)
 & =  \lim_{n \rightarrow \infty} W(U_c(m))(\mathbb{P}^x, \mathbb{P}^{y_n}) \\
& = \lim_{n \rightarrow \infty }  \min \left\{ \left.  \int U_c(m) \dee \gamma ~ \right| ~ \gamma \in \Gamma (\mathbb{P}^x, \mathbb{P}^{y_n}) \right\}.
\end{align*}
Using Theorem 5.20 of \cite{Villani08}, we know that the optimal transport plan $\pi_n$ between $\mathbb{P}^x$ and $\mathbb{P}^{y_n}$ converges, when the cost function is the continuous function $U_c(m)$, up to extraction of a subsequence, to an optimal transport plan $\pi$ for $\mathbb{P}^x$ and $\mathbb{P}^y$. This means that $\cG_c(m)(x, y_n)$ converges, up to extraction of a subsequence, to $\cG_c(m)(x,y)$.

However, the whole sequence $(\cG_c(m)(x, y_n))_{n \in \mathbb{N}}$ converges to $\cG_c(m)(x,y)$ for the following reason: assume it is not the case, since $0 \leq \cG_c(m)(x, y_n) \leq 1$ for every $n \in \mathbb{N}$,  there exists another subsequence $(z_k)_{k \in \mathbb{N}}$ such that the sequence $( \cG_c(m)(x, z_k))_{k \in \mathbb{N}}$ converges to a limit $l \neq \cG_c(m)(x, y)$. 
We also have that $\lim_{k \rightarrow \infty} z_k = y$. Hence using Theorem 5.20 of \cite{Villani08}, there is a subsequence $(a_j)_{j \in \mathbb{N}}$ of $(z_k)_{k \in \mathbb{N}}$ such that $\lim_{j \rightarrow \infty}\cG_c(m)(x, a_j) = \cG_c(m)(x, y)$, which yields a contradiction.
\end{proof}

\subsubsection{Defining our family of pseudometrics}

Lemma \ref{lemma:Gc-cont-is-cont} enables us to define for each $0 < c< 1$ a sequence of pseudometrics on $E$:
\begin{align*}
\forall x, y \in E \quad d_0^c(x,y) & = |obs(x) - obs(y)|,\\
d_{n+1}^c & = \cG_c (d_n^c).
\end{align*}
And finally we can define $\overline{d}^c = \sup_{n \in \mathbb{N}} d_n^c$ (which is also a limit since the sequence is non-decreasing using Lemma \ref{lemma:order-m-Fc-Gc}).  As a direct consequence of Lemma \ref{lemma:sup-lower-semicontinuous},  the pseudometric $\overline{d}^c$ is lower semi-continuous and is thus in the lattice $\cP$ for any $0 < c < 1$.

\subsubsection{Fixpoints}

Similarly to $\cF_c$ and $\overline{\delta}^c$,  we can show that $\overline{d}^c$ is a fixpoint for $\cG_c$. The proofs for the two next results can be found in Appendix \ref{app:diffusions}. 
Since $\cG_c(\overline{d}^c)$ is defined as the optimal transport cost when the transport function is $U_c(\overline{d}^c)$,  we need to relate $U_c(\overline{d}^c)$ to $U_c(d_n^c)$ which the next lemma does.

\begin{restatable}{lemma}{lemmaUdbarissupUdn}
\label{lemma:Ud-bar-is-sup-Ud-n}
For two trajectories $\omega, \omega'$, 
\[ U_c(\overline{d}^c)(\omega, \omega') =  \lim_{n \to \infty} U_c(d_n^c) (\omega, \omega') = \sup_n U_c(d_n^c) (\omega, \omega').  \]
\end{restatable}
We can then use this lemma to prove that $\overline{d}^c$ is a fixpoint of $\cG_c$ using Sion's minimax theorem. Lemma \ref{lemma:Ud-bar-is-sup-Ud-n} is what enables the proof of Theorem \ref{thm:m-fixpointG} to be similar to that of Theorem \ref{thm:m-fixpointFc}.

\begin{restatable}{theorem}{thmmfixpointG}
\label{thm:m-fixpointG}
The pseudometric $\overline{d}^c$ is a fixpoint of $\cG_c$.
\end{restatable}
Using Corollary \ref{cor:fixpoint-G-implies-F} and Lemma \ref{lemma:comparison-fixpointFc}, we directly obtain the following:
\begin{cor}
The pseudometric $\overline{d}^c$ is a fixpoint for $\cF_c$ for every $0 < c < 1$ and hence $\overline{\delta}^c \leq \overline{d}^c$.
\end{cor}

The following lemma is similar to Lemma \ref{lemma:comparison-fixpointFc} and the proof is essentially the same.

\begin{lemma}
\label{lemma:comparison-fixpointG}
Consider  a discount factor $0 < c < 1$ and a pseudometric $m$ in $\cP$ such that
\begin{enumerate}
\item $m$ is a fixpoint for $\cG_c$,
\item for every $x,y$, $m(x,y) \geq |obs(x) - obs(y)|$,
\end{enumerate}
then $m \geq \overline{d}^c$.
\end{lemma}

What this lemma shows is the following important result:
\begin{thm}
The pseudometric $\overline{d}^c$ is the least fixpoint of $\cG_c$ that is greater than the pseudometric $(x, y) \mapsto |obs(x) - obs(y)|$.
\end{thm}

\subsection{A real-valued logic}
\label{sec:logicG}

Similarly to Section \ref{sec:logics}, the pseudometric $\overline{d}^c$ obtained from the functional $\cG_c$ can also be described through a real-valued logic. It is worth noting that this logic needs to handle both states and trajectories and is thus defined in two parts: $\cL_\sigma$ and $\cL_\tau$.

\subsubsection{Definition of the real-valued logic}

\paragraph*{Definition of the  logics:}
This logic is defined inductively in two parts as follows:
\begin{align*}
f \in \cL_\sigma & = q~|~ obs ~|~ 1-f ~|~ \int g,\\
g \in \cL_\tau & = f \circ ev_t ~|~  \min \{ g_1, g_2\}  ~|~ \max \{ g_1, g_2\} ~|~ g \ominus q ~|~ g \oplus q,
\end{align*}
where $q \in \mathbb{Q} \cap [0,1]$, $f \in \cL_\sigma$, $g, g_1, g_2 \in \cL_\tau$ and $t \in \mathbb{Q}_{\geq 0}$. 

\paragraph*{Interpretation of the logics: }
For a fixed discount factor $0< c< 1$, the expressions in $\cL_\sigma$ are interpreted as functions $E \to [0,1]$ and expressions in $\cL_\tau$ are interpreted as functions $\Omega \to [0,1]$. The first terms of $\cL_\sigma$ are interpreted as the corresponding ones in $\grammar$: for a state $x \in E$,
\begin{align*}
q(x) & = q, \\
obs(x) & = obs(x), \\
(1-f)(x) & = 1 - f(x), \\
\left(\int g\right)(x) & = \int g ~ \dee \mathbb{P}^x.
\end{align*}
The  expressions in  $\cL_\tau$ are interpreted for a trajectory $\omega$ as
\begin{align*}
( f \circ ev_t) (\omega) & = c^t f(\omega(t)),\\
( \min \{g_1, g_2\}) (\omega)& = \min \{g_1(\omega), g_2(\omega)\},\\
( \max \{g_1, g_2\}) (\omega)& = \max \{g_1(\omega), g_2(\omega)\},\\
(g \ominus q )(\omega) & = \max \{ 0, g(\omega) - q \}, \\
(g \oplus q )(\omega) & = \min \{ 1, g(\omega) + q \}.
\end{align*}

Whenever we want to emphasize the fact that the expressions are interpreted for a certain discount factor $0<c<1$, we will write $\cL^c_\sigma$ and $\cL^c_\tau$.

\paragraph*{Additional useful expressions}
We can define additional expressions in $\cL_\sigma$ as follows:
\begin{align*}
f \ominus q & = \int [(f \circ ev_0) \ominus q],\\
\min \{ f_1, f_2\} & = \int \min \{ f_1 \circ ev_0, f_2 \circ ev_0 \},\\
\max\{ f_1, f_2 \} & = 1 - \min \{1-f_1, 1-f_2 \},\\
f \oplus q &  = 1 - ((1 - f) \ominus q).
\end{align*}
Their interpretations as functions $E \to [0,1]$ are the same as for $\grammar$.

\subsubsection{Definition of the pseudometric}


Given a fixed discount factor $0< c<1$, we can define the pseudometric $\ell^c$:
\[\ell^c (x,y) = \sup_{f \in \cL^c_\sigma} |f(x) - f(y)| = \sup_{f \in \cL^c_\sigma} (f(x) - f(y)).\]
The last equality holds since for every $f \in \cL_\sigma^c$, $\cL_\sigma^c$ also contains $1-f$.

\subsubsection{Comparison to the fixpoint metric}

In Theorem \ref{thm:metric-equal-logic-functional-F}, we proved that the real-valued logic $\grammar$ generated a pseudometric that was equal to $\overline{\delta}^c$. Such a result also holds for $\cL^c_\sigma$ and $\overline{d}^c$, but the proof is more complex as the trajectories play a bigger role. Indeed, the difficulty is now to approximate functions in $\cH(U_c(\ell^c))$ by functions in $\cL^c_\tau$.

\begin{restatable}{lemma}{lemmalogicstcontinuous}
\label{lemma:logic-st-continuous}
For any expression $f \in \cL_\sigma$, the function $x \mapsto f(x)$ is continuous and for any expression $g \in \cL_\tau$, the function $\omega \mapsto g(\omega)$ is continuous.
\end{restatable}

The proof can be found in Appendix \ref{app:diffusions}.

\paragraph*{Approximating functions on trajectories: }
We will later on use Lemma \ref{lemma:magie} to show that a function on trajectories in $\cH(U_c(\ell^c))$ can be approximated by functions in $\cL_\tau$.  We first need to look at pairs of trajectories.

\begin{restatable}{lemma}{lemmaprepmagie}
\label{lemma:prep-magie}
Consider a function $h: \Omega \to [0,1]$ and two trajectories $\omega, \omega'$ such that there exists $f$ in the logic $\cL_\sigma$ and a time $s \in \mathbb{Q}_{\geq 0}$ such that $|h(\omega) - h(\omega')| \leq c^s|f(\omega(s)) - f(\omega'(s))|$. Then for every $\delta > 0$, there exists $g$ in the logic $\cL_\tau$ such that $|h(\omega) -g(\omega)| \leq 2 \delta$ and $|h(\omega') - g(\omega')| \leq 2 \delta$.
\end{restatable}

\begin{proof}
WLOG $h(\omega) \geq h(\omega')$ and $f(\omega(s)) \geq f(\omega'(s))$ (otherwise consider $1-f$ instead of $f$).

Pick $p,q,r \in \mathbb{Q} \cap [0,1]$ such that
\begin{align*}
p & \in [c^s f(\omega'(s)) - \delta, c^s f(\omega'(s))],\\
q & \in [(h(\omega) - h(\omega') - \delta, h(\omega) - h(\omega')],\\
r & \in [ h(\omega'), h(\omega') + \delta].
\end{align*}
Define $g = ( \min \{ (f \circ ev_s) \ominus p, q \} ) \oplus r$. The details as to why $g$ works can be found in Appendix \ref{app:diffusions}.
\end{proof}

\begin{cor}
\label{cor:approx-h}
Consider a continuous function $h: \Omega \to [0,1]$ such that for any two trajectories $\omega, \omega'$, there exists $f$ in the logic and a time $s \in \mathbb{Q}$ such that $|h(\omega) - h(\omega')| \leq c^s |f(\omega(s)) - f(\omega'(s))|$. Then, for any compact set $\Omega' \subset \Omega$ the function $h$ can be approximated by functions (that are interpretations of expressions) in $\cL_\tau$ on $\Omega'$.
\end{cor}

\begin{proof}
Lemma \ref{lemma:prep-magie} says that such a function $h$ can be approximated at pairs of trajectories by functions from $\cL_\tau$.
Since  $h$ and all functions in $\cL_\tau$ are continuous, and since $\cL_\tau$ is closed under $\min$ and $\max$, we can apply Lemma \ref{lemma:magie} to get that $h$ can be approximated by functions in $\cL_\tau$ on $\Omega'$.
\end{proof}

\paragraph*{A fixpoint of $\cG$: }
Using previously proven lemmas, we can follow a proof quite similar to the one in Section \ref{sec:logics}  to show the following.

\begin{theorem}
\label{thm:mst-fixpoint}
The pseudometric $\ell^c$ is a fixpoint of $\cG_c$.
\end{theorem}

\begin{proof}
We will omit writing the index $c$ in $\ell$ and $\cG$.

We already know that $\cG(\ell) \geq \ell$ (Lemma \ref{lemma:order-m-Fc-Gc}) and we only need to prove the other inequality. Pick two states $x,y$ and recall that
\[ \cG(\ell) (x,y) = W(U_c(\ell))(\mathbb{P}^x, \mathbb{P}^y). \]
Since for any $f \in \cL_\sigma$ the map $z \mapsto f(z)$ is continuous (Lemma \ref{lemma:logic-st-continuous}) and the trajectories are continuous, it is enough to consider rational times only and we have that
\[ U_c(\ell)(\omega, \omega') = \sup_{t \in \mathbb{Q}_{\geq 0}, f \in \cL_\sigma} c^t | f(\omega(t)) - f(\omega'(t)) |. \]
Now, both the logic $\cL_\sigma$ and the set of non-negative rationals $\mathbb{Q}_{\geq 0}$ are countable, so we can enumerate their elements as $\cL_{\sigma} = (f_k)_{k \in \mathbb{N}}$ and $\mathbb{Q}_{\geq 0} = (s_k)_{k \in \mathbb{N}}$. Define the map
\[ c_k(\omega, \omega') = \max_{i, j = 1, ..., k} c^{s_i} |f_j(\omega(s_i)) - f_j(\omega'(s_i))| .\]
Note that $\lim_{k \rightarrow \infty} c_k (\omega, \omega') = U_c(\ell)(\omega, \omega')$ for any $\omega$ and $\omega'$. By Lemma \ref{lemma:wasserstein-limit},
we get that $\lim_{k \rightarrow \infty} W(c_k) (\mathbb{P}^x, \mathbb{P}^y) = W(U_c(\ell))(\mathbb{P}^x, \mathbb{P}^y)$.
It is thus enough to study $W(c_k) (\mathbb{P}^x, \mathbb{P}^y) $. Pick $k \in \mathbb{N}$.  There exists $h$ such that $\forall \omega, \omega'~ |h(\omega) - h(\omega')| \leq c_k(\omega, \omega')$ and
\[ W(c_k) (\mathbb{P}^x, \mathbb{P}^y) =\int h~ \dee \mathbb{P}^x - \int h~ \dee \mathbb{P}^y. \]
Using Lemma \ref{lemma:logic-st-continuous}, we know that $c_k$ is continuous (as a maximum of finitely many continuous functions) which means that $h$ is also continuous.  Using the second condition about $h$ and the definition of $c_k$, we know that for every pair of trajectories $\omega, \omega'$ there exists $i$ and $j \leq k$ such that $|h(\omega) - h(\omega')| \leq c^{s_i} |f_j(\omega(s_i)) - f_j(\omega'(s_i))|$. 

Since $\mathbb{P}^x$ and $\mathbb{P}^y$ are tight ($\Omega$ is Polish),  for every $\epsilon > 0$, there exists a compact set $K \subset \Omega$ such that $\mathbb{P}^x(\Omega \setminus K) \leq \epsilon / 4$ and $\mathbb{P}^y(\Omega \setminus K) \leq \epsilon / 4$.

By Corollary \ref{cor:approx-h}, there exists $(g_n)_{n \in \mathbb{N}}$ in $\cL_\tau$ that pointwise converge to $h$ on $K$.  In particular, for $n$ large enough,
\begin{align*}
\left| \int_K g_n ~ \dee \mathbb{P}^x - \int_K h ~ \dee \mathbb{P}^x \right| \leq \epsilon /4,
\end{align*}
and similarly for $\mathbb{P}^y$. We get that for $n$ large enough.
\begin{align*}
& \left| \int_\Omega g_n ~ \dee \mathbb{P}^x  - \int_\Omega h ~ \dee \mathbb{P}^x \right| \\
& \leq \left| \int_K g_n ~ \dee \mathbb{P}^x - \int_K h ~ \dee \mathbb{P}^x \right| + \int_{\Omega \setminus K} |g_n- h| ~ \dee \mathbb{P}^x \\
& \leq \epsilon / 2,
\end{align*}
and similarly for $\mathbb{P}^y$. We can thus conclude that
\begin{align*}
& W(c_k) (\mathbb{P}^x, \mathbb{P}^y) = c^t \left| \int_\Omega h ~ \dee \mathbb{P}^x - \int_\Omega h ~ \dee \mathbb{P}^y \right|\\
& \leq \left| \int_\Omega g_n ~ \dee \mathbb{P}^x - \int_\Omega h ~ \dee \mathbb{P}^x \right| + \left| \int_\Omega g_n ~ \dee \mathbb{P}^y - \int_\Omega h ~ \dee \mathbb{P}^y \right|\\
& \qquad + \left| \int_\Omega g_n ~ \dee \mathbb{P}^x - \int_\Omega g_n ~ \dee \mathbb{P}^y\right|\\
& \leq \epsilon + \ell(x, y).
\end{align*}
Since $\epsilon$ is arbitrary,  $W(c_k) (\mathbb{P}^x, \mathbb{P}^y) \leq \ell(x, y)$.

%
\end{proof}

\paragraph*{Comparison to the fixpoint metric: }
We are finally ready to compare the metric $\ell^c$ obtained through the logic $\cL_\sigma$ and $\cL_\tau$ to the metric $\overline{d}^c$ obtained through the functional $\cG_c$.

\begin{restatable}{lemma}{lemmamnf}
\label{lemma:mn-f}
For every $f \in \cL^c_\sigma$,  there exists $n \in \mathbb{N}$ such that $f \in \cH (d^c_n)$, i.e. for every states $x,y$,
\[  |f(x) - f(y)| \leq d^c_n(x,y).  \]
For every $g \in \cL^c_\tau$,  there exists $n \in \mathbb{N}$ such that $g \in \cH(U_c(d^c_n))$, i.e. for every pair of trajectories $\omega, \omega'$,
\[  |g(\omega) - g(\omega')| \leq U_c(d^c_n)(\omega, \omega').  \]
\end{restatable}

The proof can be found in Appendix \ref{app:diffusions}.

Finally, we obtain the equality of the pseudometrics defined through the fixpoint $\cG_c$ and the logics $\cL^c_\sigma$ and $\cL^c_\tau$.

\begin{theorem}
The two pseudometrics $\ell^c$ and $ \overline{d}^c$ are equal.
\end{theorem}

\begin{proof}
We know from Lemma \ref{lemma:comparison-fixpointG} and Theorem \ref{thm:mst-fixpoint} that $\ell^c \geq \overline{d}^c$.

The other direction is a direct consequence of Lemma \ref{lemma:mn-f}.
\end{proof}

\section{Conclusion}
\label{sec:discussion}

In this work, we have defined two behavioural pseudometrics for diffusions: continuous-time processes with mass conserved over time and with additional regularity condition. 

The first one $\overline{\delta}^c$ is purely based on the time-indexed family of Markov kernels and is the least fixpoint of the functional $\cF_c : m \mapsto \sup_t c^t W(m)(P_t(x), P_t(y))$ that is greater than $d_0$ (defined as $d_0 (x,y) = |obs(x) - obs(y)|$). We needed to introduce a discount factor $0<c<1$ in order to have that the image of a continuous pseudometric under $\cF_c$ is also continuous. This pseudometric corresponds to the real-valued logic $\grammar$ defined as
\begin{align*}
f \in \grammar & := q~|~ obs ~|~ \min\{ f_1, f_2 \} ~|~ 1-f ~|~ f \ominus q ~|~ \langle t \rangle f
\end{align*}
for all $f_1, f_2, f \in \grammar$, $q \in [0,1] \cap \mathbb{Q}$ and $t \in \mathbb{Q}_{ \geq 0}$ and then interpreted as a real-valued function on the state space for a given diffusion.

The second pseudometric $\overline{d}^c$ is defined on the state space but its definition is based on trajectories: it is the least fixpoint of the functional $\cG_c: m \mapsto W(U_c(m))(\mathbb{P}^x, \mathbb{P}^y)$ that is greater than $d_0$.  We also needed the discount factor $0<c< 1$ here to ensure that the image of a continuous pseudometric under $\cG_c$ is also continuous. This pseudometric $\overline{d}^c$ corresponds to the real-valued logic $\cL_\sigma$ defined inductively using another logic $\cL_\tau$ as
\begin{align*}
f \in \cL_\sigma & = q~|~ obs ~|~ 1-f ~|~ \int g\\
g \in \cL_\tau & = f \circ ev_t ~|~  \min \{ g_1, g_2\}  ~|~ \max \{ g_1, g_2\} ~|~ g \ominus q ~|~ g \oplus q
\end{align*}
where $q \in \mathbb{Q} \cap [0,1]$, $f \in \cL_\sigma$, $g, g_1, g_2 \in \cL_\tau$ and $t \in \mathbb{Q}_{\geq 0}$. The expressions in $\cL_\sigma$ are interpreted as real-valued functions on the state space $E$ whereas the expressions in $\cL_\tau$ are interpreted on the space of trajectories $\Omega$. 

While we do know that $\overline{\delta}^c \leq \overline{d}^c$, it remains to see if this inequality is in fact an equality or if there exists a counter-example proving that this is in fact a strict inequality.

Both assumptions for diffusions were essential in this study and some natural additional questions arise: ``how can we relax the assumptions we made about the processes?'', i.e. can we adapt this work to non-honest processes? What about the other regularity conditions? Could we for instance examine L\'evy processes in similar contexts?

\bibliographystyle{plain}
\bibliography{../MFPS2019/main}

\appendix

\section{Proofs for Section \ref{sec:background}}
\label{sec:proof-sec-background}

  \subsection{Lower semi-continuity}

\lemmasuplowersemicontinuous*

\begin{proof}
Pick $r \in \mathbb{R}$.  Then
\begin{align*}
f^{-1}((r, + \infty))
& = \{ x ~|~ \sup_{i \in \cI} f_i(x) > r\}\\
& = \{ x ~|~ \exists i \in \cI ~ f_i(x) > r\} = \bigcup_{i \in \cI} f_i^{-1} ((r, + \infty)).
\end{align*}
Since $f_i$ is continuous, the set $f_i^{-1}((yr + \infty))$ is open and therefore $f^{-1}((r, + \infty))$ is open which concludes the proof.
\end{proof}

    \subsection{Couplings}

In order to prove Lemma \ref{lemma:couplings-compact}, we will need to prove some more results first.

\begin{prop}
\label{prop:two}
If $P$ and $Q$ are two measures on a Polish space $X$ such that for every continuous and bounded function
$f: X \to \mathbb{R}$, we have \(\int f ~\dee P = \int f~\dee Q\), then \(P = Q\).  
\end{prop}

\begin{proof}{}
For any open set $U$, its indicator function $\mathds{1}_U$ is lower semicontinuous, which means that there exists an increasing sequence of continuous functions $f_n$ converging pointwise to $\mathds{1}_U$ (see Theorem \ref{thm:baire}).  Without loss of generality, we can assume that the functions $f_n$ are non negative (consider the sequence $\max \{ 0, f_n\}$ instead if they are not non-negative).  Using the monotone convergence theorem, we know that
\[ \lim_{n \to \infty} \int f_n ~ \dee P = \int \mathds{1}_U ~\dee P = P(U) \]
and similarly for $Q$.  By our hypothesis on $P$ and $Q$, we obtain that for every open set $U$, $P(U) = Q(U)$.

Since $P$ and $Q$ agree on the topology, they agree on the Borel algebra by Caratheodory's extension theorem.
\end{proof}

\begin{lemma}
\label{lemma:coupling-closed}
Given two probability measures $P$ and $Q$ on a Polish space $X$, the set of couplings $\Gamma(P, Q)$ is closed in the weak topology.
\end{lemma}
\begin{proof}
Consider a sequence $(\gamma_n)_{n \in \mathbb{N}}$ of couplings of $P$ and $Q$ weakly converging to a measure $\mu$ on $X \times X$, meaning that for every bounded continuous functions $f : X \times X \to \mathbb{R}$, $\lim_{n \to \infty} \int f ~ \dee \gamma_n = \int f~ \dee \mu$.  Note that $\mu$ is indeed a probability measure.  We only have to prove that the marginals of $\mu$ are $P$ and $Q$.  We will only do it for the first one, i.e.  showing that for every measurable set $A$, $\mu(A \times X) = P(A)$ since the case for $Q$ works the same.

Let $\pi:X\x X\to X$ be the first projection map.  The first marginal of $\gamma_n$ and $\mu$ are obtained as the pushforward measures $\pi_* \gamma_n$ and $\pi_*\mu$.  Indeed, $\pi_* \mu$ is defined for all measurable set $A \subset X$ as $\pi_* \mu (A) = \mu (\pi^{-1} (A)) = \mu(A \times X)$ (and similarly for the $\gamma_n$'s).  This means that for every $n$,  $\pi_* \gamma_n = P$.

For an arbitrary continuous bounded function $f: X \to \mathbb{R}$,  define the function $g: X \times X \to \mathbb{R}$ as $g = f \circ \pi$, i.e.  $g(x, y) = f(x)$.  The function $g$ is also continuous and bounded, so by weak convergence of the sequence $(\gamma_n)_{n \in \mathbb{N}}$ to $\mu$, we get that
\[ \lim_{n \to \infty} \int f \circ \pi ~ \dee  \gamma_n = \int f \circ \pi ~ \dee \mu. \]
Using the change of variables formula, we obtain that for any continuous bounded function $f$
\[ \int f ~ \dee P =  \lim_{n \to \infty} \int f ~ \dee  (\pi_*\gamma_n) = \int f ~ \dee  (\pi_*\mu). \]
Using  Proposition \ref{prop:two}, we get that $\pi_*\mu = P$.
\end{proof}

\begin{lemma}
\label{lemma:couplings-tight}
Consider two probability measures $P$ and $Q$ on Polish spaces $X$ and $Y$ respectively.   Then the set of couplings $\Gamma (P, Q)$ is tight: for all $\epsilon > 0$, there
exists a compact set $K$ in $X \times Y$ such that for all coupling $\gamma \in \Gamma(P, Q)$, $\gamma(K) > 1- \epsilon$
\end{lemma}

\begin{proof}
First note that by Ulam's tightness theorem, the probability measures $P$ and $Q$ are tight.

Consider $\epsilon > 0$.  Since $P$ and $Q$ are tight, there
exist two compact sets $K$ and $K'$ such that $P(X \setminus K) \leq \epsilon/2$ and $Q(Y \setminus K') \leq \epsilon/2$.  Define the set $C = K \times K'$.  This set is compact as a product of two compact sets.  Furthermore, for every $\gamma \in \Gamma(P,Q)$,
\begin{align*}
 \gamma((X \times Y) \setminus C) 
& = \gamma \left( [(X \setminus K) \times Y] \cup [X \times ( Y \setminus K')] \right)\\
& \leq \gamma[(X \setminus K) \times Y] + \gamma [X \times ( Y \setminus K')]\\
& = P(X \setminus K) + Q(Y \setminus K') \leq \frac{\epsilon}{2} + \frac{\epsilon}{2} = \epsilon.
\end{align*}
This shows that the set $\Gamma(P, Q)$ is tight.
\end{proof}

We can now prove Lemma \ref{lemma:couplings-compact}. Let us start by restating it.

\lemmacouplingscompact*

\begin{proof}
Using Lemma \ref{lemma:couplings-tight}, we know that the set $\Gamma(P, Q)$ is tight.
Applying Prokhorov's theorem (see Theorem \ref{thm:prokhorov}),  we get that the set $\Gamma(P, Q)$ is precompact.

Since the set $\Gamma(P, Q)$ is also closed (see Lemma \ref{lemma:coupling-closed}), then it is compact.
\end{proof}

We have used Prokhorov's theorem in the previous proof.  Here is the version cited in \cite{Villani08}
\begin{thm}
\label{thm:prokhorov}
Given a Polish space $\cX$, a subset $\cP$ of the set of probabilities on $\cX$ is precompact for the weak topology if and only if it is tight.
\end{thm}

   \subsection{Optimal transport theory}

\lemmaoptimaltransportpseudometric*

\begin{proof}
The first expression $W(c)(\mu, \nu) =  \min_{\gamma \in \Gamma(\mu, \nu)} \int c ~ \dee \gamma$ immediately gives us that $W(c)$ is 1-bounded (since we are working on the space of probability distributions on $\cX$) and that $W(c)$ is symmetric (by the change of variable formula with the function on $\cX \times \cX$, $s(x,y) = (y,x)$).

The second expression $W(c)(\mu, \nu) = \max_{h \in \cH(c)} \left| \int h ~ \dee \mu - \int h ~\dee \nu  \right|$ immediately gives us that $W(c)(\mu, \mu) = 0$ and for any three probability distributions $\mu, \nu, \theta$ and for any function $h \in \cH(c)$,
\begin{align*}
\left| \int h ~ \dee \mu - \int h ~\dee \theta \right|
& \leq \left| \int h ~ \dee \mu - \int h ~\dee \nu \right| + \left| \int h ~ \dee \nu - \int h ~\dee \theta \right|\\
& \leq W(c) (\mu, \nu) + W(c)(\nu, \theta).
\end{align*}
Since this holds for every $h \in \cH(c)$, we get the triangular inequality.
\end{proof}

\lemmawassersteinlimit*

\begin{proof}
For each $c_k$, the optimal transport cost $W(c_k)(P, Q)$ is attained for a coupling $\pi_k$.  Using Lemma \ref{lemma:couplings-compact}, we know that the space of couplings $\Gamma(P, Q)$ is compact. We can thus extract a subsequence that we will still denote by $(\pi_k)_{k \in \mathbb{N}}$ which converges weakly to some coupling $\pi \in \Gamma(P, Q)$. We will show that this coupling $\pi$ is in fact an optimal transference plan.

By Monotone Convergence Theorem,
\[ \int c ~ \dee \pi = \lim_{k \rightarrow \infty } \int c_k ~ \dee \pi. \]
Consider $\epsilon > 0$. There exists $k$ such that
\begin{align}
\label{eq:numberg}
\int c ~ \dee \pi & \leq \epsilon  \int c_k ~ \dee \pi
\end{align}

Since $\pi_k$ converges weakly to $\pi$ and since $c_k$ is continuous and bounded,
\[ \int c_k ~ \dee \pi = \lim_{n \rightarrow \infty} \int c_k ~ \dee \pi_n,\]
which implies that there exists $n_k \geq k$ such that
\begin{align}
\label{eq:numberh}
\int c_k ~ \dee \pi & \leq \epsilon +  \int c_k ~ \dee \pi_{n_k}.
\end{align}
Putting Equations (\ref{eq:numberg}) and (\ref{eq:numberh}) together, we get
\begin{align*}
\int c ~ \dee \pi
& \leq 2 \epsilon +  \int c_k ~ \dee \pi_{n_k}\\
&  \leq 2 \epsilon +  \int c_{n_k} ~ \dee \pi_{n_k} \quad \text{ since $(c_n)_{n \in \mathbb{N}}$ is increasing}\\
 & = 2 \epsilon + W(c_{n_k})(P, Q).
\end{align*}
This implies that
\[ \int c ~ \dee \pi \leq \lim_{k \rightarrow \infty} W(c_k)(P, Q). \]

The other inequality is trivial since $c_k \leq c$.
\end{proof}

\section{Details for Section \ref{sec:CTsystems}}
\label{app:CTsystems}

   \subsection{Feller-Dynkin processes}

Under the conditions of FD-semigroups, strong continuity is equivalent to the apparently weaker condition (see Lemma III.  6.7 in \cite{Rogers00a} for the proof):
\[ \forall f \in C_0(E)~~ \forall x\in E, ~~ \lim_{t\downarrow 0} (\hat{P}_t f) (x) = f(x)\]

The authors also offer the following useful extension (Theorem III.6.1):
\begin{thm}
\label{thm:Riesz}
A bounded linear functional $\phi$ on $C_0(E)$ may be written uniquely in the form
\[ \phi(f) = \int_E f(x)~ \mu(\dee x) \]
where $\mu$ is a signed measure on $E$ of finite total variation.
\end{thm}
As stated in the Riesz representation theorem, the above measure $\mu$ is inner regular.   Theorem \ref{thm:Riesz} is also known as the Riesz-Markov-Kakutani representation theorem in some other references.

Theorem \ref{thm:Riesz} has the following corollary (Theorem III.6.2) where $b\mathcal{E}$ denotes the set of bounded, $\cE$-measurable functions $E \to \mathbb{R}$.
\begin{restatable}{corollary}{cormarkovfromsemigroup}
\label{cor:markov-from-semigroup}
Suppose that $V: C_0(E) \to b\mathcal{E}$ is a (bounded) linear operator that is sub-Markov in the sense that $0 \leq f \leq 1$ implies $0 \leq Vf \leq 1$.  Then there exists a unique sub-Markov kernel (also denoted by) $V$ on $(E, \mathcal{E})$ such that for all $f \in C_0(E)$ and $x \in E$
\[ Vf(x) = \int f(y) ~ V(x, \dee y). \]
\end{restatable}

While the proof is left as an exercise in \cite{Rogers00a}, we choose to explicitly write it down for clarity


\begin{proof}
For every $x$ in $E$, write $V_x$ for the functional $V_x(f) = Vf(x)$.  This functional is bounded and linear which enables us to use Theorem \ref{thm:Riesz}: there exists a signed measure $\mu_x$ on $E$ of finite total variation such that
\[ V_x (f) = \int_E f(y) ~\mu_x(\dee y) \]
We claim that $V: (x, B) \mapsto \mu_x(B)$ is a sub-Markov kernel, i.e.
\begin{itemize}
\item For all $x$ in $E$, $V(x, -)$ is a subprobability measure on $(E, \mathcal{E})$
\item For all $B$ in $\mathcal{E}$, $V( -, B)$ is $\mathcal{E}$-measurable.
\end{itemize}

The first condition directly follows from the definition of $V$: $V(x, -) = \mu_x$ which is a measure and furthermore $V(x, E) = \mu_x(E) = V_x (1)$ (where $1$ is the constant function over the whole space $E$ which value is $1$).  Using the hypothesis that $V$ is Markov, we get that $V 1 \leq 1$.  We have to be more careful in order to prove that $V(x, B) \geq 0$ for every measurable set $B$: this is a consequence of the regularity of the measure $\mu_x$ (see Proposition 11 of section 21.4 of \cite{Royden10}).  This shows that $\mu_x$ is a subprobability measure on $(E, \mathcal{E})$.

Now, we have to prove that for every $B \in \mathcal{E}$, $V(-, B)$ is measurable.  Recall that the set $E$ is $\sigma$-compact: there exists countably many compact sets $K_k$ such that $E = \bigcup_{k \in \mathbb{N}} K_k$.  For $n \in \mathbb{N}$, define $E_n =  \bigcup_{k = 0}^n K_k$ and $B_n = B \cap E_n$.

 For every $n \in \mathbb{N}$, there exists a sequence of  functions $(f_j^n)_{j \in \mathbb{N}} \subset C_0(E)$  that converges pointwise to $\mathds{1}_{B_n}$, i.e.  for every $x \in E$, $\mu_x(B_n) = \lim_{j \to +\infty} Vf_j^n(x)$.  Since the operator $V: C_0(E) \to b\cE$, the maps $Vf_j^n$ are measurable which means that $V(-, B_n): x \mapsto \mu_x(B_n)$ is measurable.

Since for every $x \in E$, $\mu_x(B) = \lim_{n \to +\infty} \mu_x(B_n)$, this further means that $V(-, B): x \mapsto \mu_x(B)$ is measurable.
\end{proof}

We can then use this corollary to derive Proposition \ref{prop:Riesz-use} which relates these FD-semigroups with Markov
kernels.  This allows one to see the connection with familiar
probabilistic transition systems.  

   \subsection{Diffusions}
   
\begin{lemma}
\label{lemma:COtopology-metrized}
The topology $\cT$ on $\Omega$ is metrized by $U_c(\Delta)$ (for any $0 < c< 1$).
\end{lemma}

\begin{proof}
The compact-open topology $\cT$ also corresponds to the topology of uniform convergence on every set $[0, T]$ for $T \geq 0$. Let us show that it corresponds to convergence according to $U_c(\Delta)$.

Let the sequence $(\omega_n)_{n \in \mathbb{N}}$ uniformly converge on every set $[0, T]$ for $T \geq 0$. Let us note $\omega$ for its limit. Pick $\epsilon > 0$. Since $c < 1$,  there exists some $T$ such that $c^T \leq \epsilon$. Then there exists $N$ such that
\[ \forall n \geq N ~\forall t \in [0, T] ~ \Delta (\omega_n(t), \omega(t)) \leq \epsilon.\]
Furthermore, $c^t \leq 1$ and hence 
\[ \forall n \geq N ~\forall t \in [0, T] ~ c^t \Delta (\omega_n(t), \omega(t)) \leq \epsilon.\]
For every $t \geq T$, we have $c^t \Delta (\omega_n(t), \omega(t)) \leq c^T \leq \epsilon$. We thus obtain that there exists $N$ such that
\[ \forall n \geq N ~\forall t \geq 0 ~ c^t \Delta (\omega_n(t), \omega(t)) \leq \epsilon,\]
and thus $\lim_{n \rightarrow \infty} U_c(\Delta)(\omega_n, \omega) = 0$.

For the other direction,  let the sequence $(\omega_n)_{n \in \mathbb{N}}$  converge to $\omega$ according to $U_c(\Delta)$ and let $T \geq 0$ and $\epsilon > 0$. There exists $N$ such that for all $n \geq N$, $U_c(\Delta) (\omega_n, \omega) \leq c^T \epsilon$.  This means in particular that for every time $t \geq 0$, $c^t \Delta (\omega_n(t), \omega(t)) \leq c^T \epsilon$.  We can thus restrict the times considered and get:
\[ \forall t \in [0, T] \Delta (\omega_n(t), \omega(t)) \leq c^{T - t} \epsilon \leq \epsilon. \]
The last inequality holds since $T - t \geq 0$ and $c < 1$.
\end{proof}

\begin{lemma}
\label{lemma:Omega-Polish}
The space $\Omega$ equipped with the compact-open topology is Polish.
\end{lemma}

\begin{proof}
First let us show completeness of $(\Omega, U_c(\Delta))$. Let $(\omega_n)_{n \geq 0}$ be a Cauchy sequence.  As discussed in the proof of Lemma \ref{lemma:COtopology-metrized},  $(\omega_n)_{n \geq 0}$ is Cauchy under the uniform metric over any compact interval. This means that there exists some $\tilde{\omega}: [0, \infty) \to E$ such that $(\omega_n)_{n \geq 0}$ converges uniformly to $\tilde{\omega}$ on every compact $[0, T]$. This also implies that $\omega$ is continuous on every compact $[0, T]$ and thus on $[0, \infty)$.

Now let us show separability of $\Omega$ equipped with the compact-open topology. Since for a fixed $T$, the set of continuous functions $[0, T] \to E$ is separable (it is Polish), it is enough to show that the set $\Omega_0$ is dense in $\Omega$ wrt $U_c(\Delta)$ where
\[ \Omega_0 = \{ \omega ~|~ \exists T \geq 0 ~ \forall t \geq T ~~ \omega(t) = \omega(T) \}. \]
For a trajectory $\omega \in \Omega$, define $\omega_k \in \Omega_0$ for all $k \in \mathbb{N}$ by
\[ \omega_k(t) = \begin{cases}
\omega(t) ~~\text{if } t \leq k\\
\omega(k)~~\text{if } k \leq t\\
\end{cases} \]
Now, for any $k \in \mathbb{N}$, 
\[U_c(\Delta)(\omega, \omega_k) = \sup_{t \geq k} c^t \Delta (\omega(t), \omega(k)) \leq c^k\]
since $\Delta$ is 1-bounded. Furthermore, since $c< 1$, we get that $\Omega_0$ is dense in $\Omega$.
\end{proof}

\section{Proof for Section \ref{sec:functionals-metric}}
\label{app:functionals-metric}

\lemmaordermFc*

\begin{proof}
Consider a pair of states $x,y$. Then
\begin{align*}
\cF_c(m)(x,y)
& = \sup_{t \geq 0} c^t W(m)(P_t(x), P_t(y))\\
& \geq W(m)(P_0(x), P_0(y))\\
& = \inf_{\gamma \in \Gamma(P_0(x), P_0(y))} \int m ~ \dee \gamma.
\end{align*}
Since $P_0(x)$ is the dirac distribution at $x$ and similarly for $P_0(y)$, the only coupling $\gamma$ between $P_0(x)$ and $P_0(y)$ is the product measure $P_0(x) \times P_0(y)$ and thus $W(m)(P_0(x), P_0(y)) = m(x, y)$, which concludes the proof.
\end{proof}

\section{Proofs for Section \ref{sec:logics}}
\label{app:sec-logics}

\lemmaalphaf*

\begin{proof}
This proof is done by induction on the structure of $f$:
\begin{itemize}
\item If $f = q$, then $f(x) - f(y) = 0 \leq \delta^c_0(x,y)$.
\item If $f = obs$, then $f(x) - f(y) = obs(x) - obs(y) \leq  \delta^c_0(x,y)$.
\item If $f = 1 - g$ and there exists $n$ such that for every $x,y$, $ g(x) - g(y) \leq \delta^c_n(x,y)$, then $f(x) - f(y) = g(y) - g(x) \leq \delta^c_n(y,x) =  \delta^c_n(x,y)$.
\item If $f = g \ominus q$  and there exists $n$ such that for every $x,y$, $ g(x) - g(y) \leq \delta^c_n(x,y)$,  then it is enough to study the case when $f(x) = g(x) - q \geq 0$ and $f(y) = 0 \geq g(y) -q$. In that case,
\begin{align*}
f(x) - f(y) & = g(x) - q \leq g(x) - q - (g(y) - q) = g(x) - g(y) \leq \delta_n^c(x,y)\\
f(y) - f(x) &=  q - g(x) \leq 0 \leq \delta^c_n(x,y).
\end{align*}
\item If $f = \min \{ f_1, f_2\}$ and for $i = 1,2$ there exists $n_i$ such that for every $x,y$, $ f_i(x) - f_i(y) \leq \delta^c_{n_i}(x,y)$. Then write $n = \max\{ n_1, n_2\}$. There really is only one case to consider: $f(x) = f_1(x) \leq f_2(x)$ and $f(y) = f_2(y) \leq f_1(y)$. In that case
\[ f(x) - f(y) = f_1(x) - f_2(y) \leq f_2(x) - f_2(y) \leq \delta^c_{n_2}(x,y) \leq \delta^c_n(x,y). \]
\item If $f = \langle t \rangle g$ and there exists $n$ such that for every $x,y$, $ g(x) - g(y) \leq \delta^c_n(x,y)$,  then
\begin{align*}
f(x) - f(y)
& = c^t \hat{P}_t g(x) - c^t \hat{P}_t g(y) = c^t \left( \hat{P}_t g(x) -\hat{P}_t g(y) \right)\\
& \leq c^t W(\delta_n^c)(P_t(x), P_t(y))\\
& \leq \delta_{n+1}^c (x,y).
\end{align*}
\end{itemize}
\end{proof}

\lemmaprepmagieF*

\begin{proof}
WLOG $h(z) \geq h(z')$ and $f(z) \geq f(z')$ (otherwise consider $1-f$ instead of $f$).

Pick $p,q,r \in \mathbb{Q} \cap [0,1]$ such that
\begin{align*}
p & \in [f(z') - \delta, f(z')],\\
q & \in [h(z) - h(z') - \delta, h(z) - h(z')],\\
r & \in [h(z'), h(z') + \delta].
\end{align*}
Define $g = ( \min \{ f \ominus p, q \} ) \oplus r$. Then,
\begin{align*}
(f \ominus p)(z)
& \in [f(z) - f(z'), f(z) - f(z') + \delta] ~~ \text{since } f(z) \geq f(z'), \\
(\min \{ f \ominus p, q \})(z)
& = q ~~ \text{since } q \leq h(z) - h(z') \leq f(z) - f(z'),\\
g(z) & = \min \{ 1, q + r\} \in [h(z) - \delta, h(z) + \delta)] ~~ \text{as } h(z) \leq 1,
\end{align*}
which means that $|h(z) - g(z)| \leq  \delta$.
\begin{align*}
(f \ominus p)(z')
& = \max \{ 0, f(z') - p\} \in [0, \delta],\\
(\min \{ f \ominus p, q \})(z')
& \in [0, \delta],\\
g(z')& \in [h(z'), h(z') +2 \delta],
\end{align*}
which means that $|h(z') - g(z')| \leq 2 \delta$.
\end{proof}

\section{Proofs for Section \ref{sec:diffusions}}
\label{app:diffusions}

\lemmaordermFcGc*

\begin{proof}
The first inequality corresponds to Lemma \ref{lemma:order-m-Fc}.

For the second inequality: pick $t \geq 0$ and $h: E \to [0,1] \in \cH(m)$. Define $h_t : \Omega \to [0,1]$ as $h_t(\omega) = c^t h(\omega(t))$. Then note that for any trajectories $\omega, \omega'$,
\begin{align*}
|h_t(\omega) - h_t(\omega')|
& = c^t |h(\omega(t)) - h(\omega'(t))|\\
& \leq c^t m(\omega(t), \omega'(t))\\
& \leq U_c(m) (\omega, \omega').
\end{align*}
This means that $h_t \in \cH(U_c(m))$ and thus
\begin{align*}
\cG_c(m)(x,y)
& \geq \int h_t ~ \dee \mathbb{P}^x - \int h_t ~ \dee \mathbb{P}^y\\
& = c^t \left(\int h ~ \dee P_t(x) - \int h ~ \dee P_t(y) \right)\\
& = c^t \left( \hat{P}_t h(x) - \hat{P}_t h(y) \right),
\end{align*}
which allows us to conclude since it holds for every time $t \geq 0$.
\end{proof}

\corfixpointGimpliesF*

\begin{proof}
For a metric $m \in \cP$ that is a fixpoint of $\cG_c$ and two states $x$ and $y$:
\[ m(x, y) \leq \cF_c(m) (x,y) \leq \cG_c(m)(x,y) = m(x,y), \]
which means that $m = \cF_c(m)$.
\end{proof}

\lemmaUdbarissupUdn*

\begin{proof}
We will omit writing $c$ as an index for the pseudometrics $d_n$ and $\overline{d}$ throughout this proof. 

First note that $U_c(d_{n+1}) \geq U_c(d_n)$ since $U_c$ is monotone which shows the second equality.
Besides,
\begin{align*}
U_c(\overline{d})(\omega, \omega')
&= \sup_{t \geq 0} c^t \overline{d} (\omega(t), \omega'(t)) \\
&= \sup_{t \geq 0} c^t \sup_n d_n (\omega(t), \omega'(t)) \\
& = \sup_n \sup_{t \geq 0}  c^t d_n (\omega(t), \omega'(t)) \\
& = \sup_n U_c(d_n) (\omega, \omega').
\end{align*}
\end{proof}

\thmmfixpointG*

\begin{proof}
We will omit writing $c$ as an index for the pseudometrics $d_n$ and $\overline{d}$ throughout this proof.  Fix two states $x,y$.

The space of finite measures on $\Omega \times \Omega$ is a linear topological space.
Using Lemma \ref{lemma:couplings-compact}, we know that the set of couplings $\Gamma (\mathbb{P}^x, \mathbb{P}^y)$ is a compact subset.  Besides, it is also convex.

The space of bounded pseudometrics on $E$ is also a linear topological space.  We have defined a sequence $d_0, d_1, ...$ on that space.  We define $Y$ to be the set of linear combinations of pseudometrics $\sum_{n \in \mathbb{N}} a_n d_n$ such that for every $n$, $a_n \geq 0$ (and finitely many are non-zero) and $\sum_{n \in \mathbb{N}} a_n = 1$.  This set $Y$ is convex.

Define the function
\begin{align*}
\Xi : \Gamma(\mathbb{P}^x, \mathbb{P}^y) \times Y & \to [0,1]\\
(\gamma, m) & \mapsto\int U_c(m) ~ \dee \gamma.
\end{align*}

For $\gamma \in \Gamma(\mathbb{P}^x, \mathbb{P}^y)$, the map $\Xi(\gamma, \cdot)$ is continuous by dominated convergence theorem and it is monotone and hence quasiconcave.  For a given $m \in Y$, by definition of the Lebesgue integral, $\Xi(\cdot, m)$ is continuous and linear.

We can therefore apply Sion's minimax theorem:
\begin{align}
\label{eq:minimaxG}
\min_{\gamma \in \Gamma(\mathbb{P}^x, \mathbb{P}^y)} \sup_{m \in Y} \int U_c(m) ~\dee\gamma =  \sup_{m \in Y} \min_{\gamma \in \Gamma(\mathbb{P}^x, \mathbb{P}^y)}  \int U_c(m) ~ \dee \gamma.
\end{align}

Similarly to what was shown in the proof of Theorem \ref{thm:m-fixpointFc}, for an arbitrary functional $\Psi$ on $Y$ such that $m \leq m' \Rightarrow \Psi(m) \leq \Psi(m')$, 
\begin{align}
\label{eq:numberf}
\sup_{m \in Y} \Psi(m) &= \sup_{n \in \mathbb{N}} \Psi (d_n).
\end{align}
The right-hand side of Equation (\ref{eq:minimaxG}) is
\begin{align*}
& \sup_{m \in Y} \min_{\gamma \in \Gamma(\mathbb{P}^x, \mathbb{P}^y)}  \int U_c(m) ~ \dee \gamma = \sup_{m \in Y} \cG_c(m)(x,y)\\
& = \sup_{n \in \mathbb{N}} \cG_c(d_n)(x,y) \quad \text{(previous result in Equation (\ref{eq:numberf}))}\\
& = \sup_{n \in \mathbb{N}} d_{n+1}(x,y) \quad \text{(definition of } (d_n)_{n \in \mathbb{N}})\\
& =  \overline{d}(x,y).
\end{align*}

The left-hand side of Equation (\ref{eq:minimaxG}) is
\begin{align*}
& \min_{\gamma \in \Gamma(\mathbb{P}^x, \mathbb{P}^y)} \sup_{m \in Y} \int U_c(m) ~\dee \gamma\\
& = \min_{\gamma \in \Gamma(\mathbb{P}^x, \mathbb{P}^y)} \sup_{n \in \mathbb{N}} \int U_c(d_n) ~ \dee \gamma \quad \text{(previous result in Equation (\ref{eq:numberf}))}\\
& = \min_{\gamma \in \Gamma(\mathbb{P}^x, \mathbb{P}^y)}  \int \sup_{n \in \mathbb{N}} U_c(d_n) ~ \dee \gamma \quad \text{(dominated convergence theorem)}\\
& = \min_{\gamma \in \Gamma(\mathbb{P}^x, \mathbb{P}^y)}  \int U_c( \sup_{n \in \mathbb{N}} d_n) ~ \dee \gamma ~~\text{(Lemma \ref{lemma:Ud-bar-is-sup-Ud-n})}\\
& = \min_{\gamma \in \Gamma(\mathbb{P}^x, \mathbb{P}^y)}  \int U_c( \overline{d}) ~ \dee \gamma  = \cG_c(\overline{d})(x,y).
\end{align*}

This concludes the proof since we get that $\cG_c(\overline{d})(x,y) = \overline{d}(x,y)$ for every two states $x$ and $y$.
\end{proof}

\lemmalogicstcontinuous*

\begin{proof}
This is done by induction on the structure of $f$ and $g$. There are two cases that are not straightforward.
\begin{itemize}
\item First, if $f = \int g$ with $g \in \cL_\tau$, then, since $x \mapsto \mathbb{P}^x$ is continuous and $g$ is continuous and bounded in $\Omega$, we get that $x \mapsto \int g ~ \dee \mathbb{P}^x$ is continuous in $E$.
\item Second, if $g =  f \circ ev_t$ with $f \in \cL_\sigma$ and $t \in \mathbb{Q}_{\geq 0}$, then consider $(\omega_n)_{n \in \mathbb{N}}$ that converges to $\omega$ according to $U_c(\Delta)$:
\[ \lim_{n \rightarrow \infty}U_c(\Delta)(\omega_n, \omega) = 0\]
i.e. $  \lim_{n \rightarrow \infty} \sup_{s \geq 0} c^s \Delta (\omega_n (s), \omega(s)) = 0$. In particular, this entails that for our particular time $t$,
\[ \lim_{n \rightarrow \infty} \Delta (\omega_n (t), \omega(t)) = 0 .\]
And since $f$ is continuous by induction hypothesis,
\[ |g (\omega_n) - g(\omega)| = c^t |f (\omega_n(t)) - f (\omega(t))| \underset{n \rightarrow\infty}{\longrightarrow} 0. \]
\end{itemize}
\end{proof}

\lemmaprepmagie*

\begin{proof}
WLOG $h(\omega) \geq h(\omega')$ and $f(\omega(s)) \geq f(\omega'(s))$ (otherwise consider $1-f$ instead of $f$).

Pick $p,q,r \in \mathbb{Q} \cap [0,1]$ such that
\begin{align*}
p & \in [c^s f(\omega'(s)) - \delta, c^s f(\omega'(s))],\\
q & \in [(h(\omega) - h(\omega') - \delta, h(\omega) - h(\omega')],\\
r & \in [ h(\omega'), h(\omega') + \delta].
\end{align*}
Define $g = ( \min \{ (f \circ ev_s) \ominus p, q \} ) \oplus r$. Then,
\begin{align*}
(f \circ ev_s)(\omega) & = c^s f(\omega(s)),\\
((f \circ ev_s) \ominus p)(\omega) & = \max \{ 0, c^s f(\omega(s)) - p\},
\end{align*}
i.e.  $((f \circ ev_s) \ominus p)(\omega) \in [c^s [f(\omega(s)) - f(\omega'(s))], c^s [f(\omega(s)) - f(\omega'(s))]+ \delta] $ since  $f(\omega(s)) \geq f(\omega'(s))$.
\[
(\min \{ (f \circ ev_s) \ominus p, q \})(\omega) = q\]
since $q \leq  h(\omega) - h(\omega')\leq c^s [f(\omega(s)) - f(\omega'(s))]$.
\begin{align*}
q + r& \in [h(\omega) - \delta, h(\omega) + \delta],\\
g(\omega) & = \min \{ 1, q + r\} \in [h(\omega) - \delta,  h(\omega) + \delta] ~~ \text{as } h(\omega) \leq 1.
\end{align*}
Meanwhile,
\begin{align*}
(f \circ ev_s)(\omega') & = c^s f(\omega'(s)),\\
((f \circ ev_s) \ominus p)(\omega')
& = \max \{ 0, c^s f(\omega'(s)) - p\} \in [0, \delta],\\
(\min \{ (f \circ ev_s) \ominus p, q \})(\omega')
& \in [0, \delta],\\
g(\omega') & = \min \{ 1,  (\min \{ (f \circ ev_s) \ominus p, q \})(\omega') + r\}\\
& \in [h(\omega'), h(\omega') + 2 \delta].
\end{align*}
\end{proof}

\lemmamnf*

\begin{proof}
The proof is done by induction on the structure of the logic $\cL_\sigma$ and $\cL_\tau$ and quite similar to that of Lemma \ref{lemma:alpha-f}. We will only mention the different cases here.
\begin{itemize}
\item If $f = \int g$ with $g \in \cH(U_c(d^c_n))$, then
\begin{align*}
|f(x) - f(y)|
& = \left| \int g ~ \dee  \mathbb{P}^x - \int g~ \dee \mathbb{P}^y \right|\\
& \leq W(U_c(d^c_n)) (\mathbb{P}^x, \mathbb{P}^y)\\
& = \cG_c (d^c_n) (x, y) = d_{n+1}^c(x,y).
\end{align*}
\item If $g = \min \{ g_1, g_2\}$ with $g_i \in \cH(U_c(d_{n_i}^c))$, then for any two trajectories $\omega, \omega'$, there really is only one case to consider: $g_1(\omega) \leq g_2 (\omega)$ and $g_2(\omega') \leq g_1 (\omega')$. In that case,
\begin{align*}
|g(\omega) - g(\omega')|
& = |g_1(\omega) - g_2(\omega')|\\
& = \max \{g_1(\omega) - g_2(\omega'), g_2(\omega') - g_1(\omega) \}\\
&  \leq \max \{g_2(\omega) - g_2(\omega'), g_1(\omega') - g_1(\omega) \}\\
& \leq \max \{ U_c(d_{n_2}^c)(\omega, \omega'), U_c(d_{n_1}^c)(\omega, \omega')\}\\
& \leq U_c(d_{n}^c)(\omega, \omega')
\end{align*}
with $n = \max \{n_1, n_2\}$, which means that $g \in \cH(U_c(d_{n}^c))$.
\item If $g = f \circ ev_t$ with $f \in \cH(d_n^c)$, then for any two trajectories $\omega, \omega'$,
\begin{align*}
|g(\omega) - g(\omega')|
& =  \left|c^t f(\omega(t)) - c^t f(\omega'(t))\right|\\
& \leq c^t d_n^c (\omega(t), \omega'(t))\\
& \leq U_c(d_n^c)(\omega, \omega').
\end{align*}
\end{itemize}
\end{proof}

\end{document}